\documentclass[a4paper,UKenglish,cleveref, autoref, thm-restate]{lipics-v2021}

\pdfoutput=1 
\hideLIPIcs

\newcommand{\bettersc}[1]{%
	\ifdim\fontdimen1\font>0pt%
		\textsl{\textsc{#1}}%
	\else%
		\textsc{#1}%
	\fi%
}

\usepackage[xcolor, hyperref, notion, quotation]{knowledge}

\knowledge{notion}
	| string synchronizing set
	| string synchronizing sets
	| String synchronizing set
	| String synchronizing sets
	| synchronizing set
	| synchronizing sets
	| Synchronizing set
	| Synchronizing sets
\knowledge{notion}
	| common CRCW PRAM
	| Common CRCW PRAM
	| common CRCW PRAMs
	| PRAMs
	| PRAM
\knowledge{notion}
	| LCE
	| longest common extension
\knowledge{notion}
	| string synchronizing sets hierarchy
	| synchronizing sets hierarchy
\knowledge{notion}
	| hierarchy of decompositions
	| decompositions hierarchy
\knowledge{notion}
	| Consistency
	| consistency
\knowledge{notion}
	| Density
	| density
\knowledge{notion}
	| Sparseness
	| sparseness
\knowledge{notion}
	| Names
	| names
	| consistent names
\knowledge{notion}
	| Dyck language
	| Dyck languages
	| Dyck Languages
\knowledge{notion}
	| square
	| Square
	| squares
	| Squares
\knowledge{notion}
	| square-free
	| square-freeness
\knowledge{notion}
	| monoid
	| monoids
\knowledge{notion}
	| affected segment
	| affected segments
\knowledge{notion}
	| relevant segment
	| relevant segments
\knowledge{notion}
	| affected path
\knowledge{notion}
	| deactivates factors
	| deactivated factors
	| deactivated factor
	| deactivates a factor
	| deactivates
	| deactivated
	| Deactivating
	| Deactivated factors
\knowledge{notion}
	| void characters
	| void character
	| void-characters
	| void-character
\knowledge{notion}
	| erasing-function
\knowledge{notion}
	| effect-preserving
\knowledge{notion}
	| effect-inducing
\knowledge{notion}
	| current instance
\knowledge{notion}
	| future instance

\knowledgenewrobustcmd \calS {\mathrel {\cmdkl {\mathcal{S}}}}	
\knowledgenewrobustcmd \calR {\mathrel {\cmdkl {\mathcal{R}}}}

\knowledgenewrobustcmd \calBtau {\mathrel {\cmdkl {\mathcal{B}^{\tau}}}}
\knowledgenewrobustcmd \calNBtau {\mathrel {\cmdkl {\mathcal{N}_B^{\tau}}}}

\knowledgenewrobustcmd \calSl {\mathrel {\cmdkl {\mathcal{S}^{\ell}}}}
\knowledgenewrobustcmd \calDl {\mathrel {\cmdkl {\mathcal{D}^{\ell}}}}
\knowledgenewrobustcmd \calNDl {\mathrel {\cmdkl {\mathcal{N}_D^{\ell}}}}

\knowledgenewrobustcmd \Dell {\mathrel {\cmdkl {D^{\ell}}}}

\knowledgenewrobustcmd \erase [1]{\ensuremath{\cmdkl {e_{#1}}}}

\knowledgenewrobustcmd \setcount [2]{\cmdkl {| }#1\cmdkl {| }_{#2}}

\knowledgenewrobustcmd \dynamicLCE {\cmdkl{\bettersc{DynamicLce}}}
\knowledgenewrobustcmd \dynamicMarkedString {\cmdkl{\bettersc{DynamicMarkedString}}}
\knowledgenewrobustcmd \dynamicNames {\cmdkl{\bettersc{DynamicNames}}}
\knowledgenewrobustcmd \dynamicString {\cmdkl{\bettersc{DynamicString}}}
\knowledgenewrobustcmd \dynamicNeighbor {\cmdkl{\bettersc{DynamicNeighbor}}}
\knowledgenewrobustcmd \dynamicFunction {\cmdkl{\bettersc{DynamicFunction}}}
\knowledgenewrobustcmd \member [1]{\cmdkl{\bettersc{Member}(#1)}}
\knowledgenewrobustcmd \stringEquality {\cmdkl{\bettersc{StringEquality}}}
\knowledgenewrobustcmd \rangeEval [1]{\cmdkl{\bettersc{RangeEval}(#1)}}
\knowledgenewrobustcmd \dynamicSquareFreeness {\cmdkl{\bettersc{DynamicSquareFreeness}}}
\knowledgenewrobustcmd \dynamicRangeSquareFreeness {\cmdkl{\bettersc{DynamicRangeSquareFreeness}}}
\knowledgenewrobustcmd \dynamicStringMatching {\cmdkl{\bettersc{Dynamic2-1StringMatching}}}
\knowledgenewrobustcmd \dynamicPrefixSuffix {\cmdkl{\bettersc{DynamicPrefixSuffix}}}

\title{Longest Common Extension of a Dynamic String in Parallel Constant Time}

\titlerunning{Dynamic LCE in Constant Time}

\author{Daniel Alexander {Albert}}{TU Dortmund University, Germany }{daniel.albert@tu-dortmund.de}{}{}

\authorrunning{D.\ Albert}

\Copyright{Daniel Albert}



\nolinenumbers

\begin{document}

\maketitle

\begin{abstract}
	A "longest common extension" ("LCE") query on a string computes the length of the longest common suffix or prefix at two given positions.
	A dynamic "LCE" algorithm maintains a data structure that allows efficient "LCE" queries on a string that can change via character insertions and deletions.
	
	A dynamic parallel constant-time algorithm is presented that can maintain "LCE" queries on a "common CRCW PRAM" with $\mathcal{O}(n^{\epsilon})$ work, for any $\epsilon > 0$.
	The algorithm maintains a "string synchronizing sets hierarchy", which it uses to answer substring equality queries, which it in turn uses to answer "LCE" queries.
	To achieve constant runtime, the algorithm allows parts of its information to become outdated by up to $\log n \log^* n$ updates.
	It answers queries by combining this slightly outdated information with a list of the recent changes.
	
	Two applications of this dynamic "LCE" algorithm are shown.
	Firstly, a dynamic parallel constant-time algorithm can maintain membership in a "Dyck language" $D_k, k > 0$ with $\mathcal{O}(n^{\epsilon})$ work for any $\epsilon > 0$.
	Secondly, a dynamic parallel constant-time algorithm can maintain "squares" with $\mathcal{O}(n^{\epsilon})$ work for any $\epsilon > 0$.
\end{abstract}

\section{Introduction}
\label{sec:introduction}

Answering "longest common extension" ("LCE") queries is a fundamental problem that appears as part of many string algorithms.
An ""LCE"" query determines, for two given indices $i,j$ of a string $S$, the length of the longest common prefix (\text{LCP}) of $S[i,|S|]$ and $S[j,|S|]$, or the length of the longest common suffix (\text{LCS}) of $S[1,i]$ and $S[1,j]$.

\subparagraph{Background on LCE Queries}
"LCE" queries trace back to Landau and Vishkin \cite[Section 2.3]{landauFastStringMatching1988}, who use the name \text{MAX-LENGTH} for what we call the length of the \text{LCP}.
Optimal "LCE" algorithms build a data structure for a given string which answers "LCE" queries in constant time.
The first construction algorithm runs in linear time for constant alphabets (\cite[Section 4]{landauFastStringMatching1988} and \cite{weinerLinearPatternMatching1973}).
Later algorithms achieve the same runtime even for integer alphabets \cite{farach-coltonSortingcomplexitySuffixTree2000}.

These algorithms are sequential and run on a single processor.
With parallel machines, we can find algorithms with shorter runtime at the cost of adding more processors.

An early parallel "LCE" algorithm is the skew algorithm by Kärkkäinen and Sanders \cite[Theorem 2]{karkkainenSimpleLinearWork2003}, which runs in $\mathcal{O}(\log^2 n)$ time on $\mathcal{O}(n/ \log n)$ processors.
This algorithm runs on the very restrictive \text{EREW PRAM} model.
Shun improved this algorithm to run on $\mathcal{O}(n/\log^2 n)$ processors \cite[Theorem 3.6]{shunFastParallelComputation2014}.
This algorithm is work-efficient, meaning that simulating it with just one processor is as efficient as the best known sequential algorithm.

These algorithms operate in the static setting.
In the dynamic setting, we are looking for a data structure which can keep up with changes to the string.
We are interested in dynamic problems where we can insert or delete single characters anywhere in the string.
There has been some research on dynamic "LCE" algorithms on sequential machines.

Early work on dynamic strings was done by Gu, Farach and Beigel \cite{guEfficientAlgorithmDynamic1994}.
They studied the related dynamic text indexing problem.
This problem is about maintaining a data structure for a changing string which can report all occurrences of a given pattern in the string.
Gawrychowski et al.\ show an optimal data structure for a different related problem where updates and "LCE" queries run in logarithmic time, with high probability \cite{gawrychowskiOptimalDynamicStrings2018}.%
\footnote{%
	Gawrychowski et al.\ study a different but related problem and their algorithm is more powerful.
	The dynamic "LCE" problem is reducible to the problem in Gawrychowski et al.
	This reduction increases the runtime of queries from constant to logarithmic.%
}

\subparagraph{Parallel Constant-Time Algorithms}
We study the problem of dynamic "LCE" queries in the context of parallel constant-time algorithms.
We use the parallel computational model described by Shiloach and Vishkin \cite{shiloachFindingMaximumMerging1981}, which since has become known as the "common CRCW PRAM" (see e.g.\ \cite{fichRelationsConcurrentwriteModels1984}).
A "common CRCW PRAM" employs multiple parallel processors that access a shared memory.
Multiple processors can access the same cell of memory concurrently, but concurrent writes are only allowed if all write the same value.

We study parallel algorithms where the maximum runtime over all processors is constant.
Our measure for the complexity of such algorithms is the number of processors.

\subparagraph{Our Contributions}
Our primary result is the following theorem.

\begin{theorem}\label{the:LCE}
	For any $\epsilon > 0$, there is a dynamic algorithm on a "common CRCW PRAM" that maintains a data structure for a string $S$ with fixed maximum size $n$, where
	\begin{itemize}
		\item the size of the data structure is in $\mathcal{O}(n \log n \log^* n)$,\footnote{
			The iterated logarithm $\log^* n$ counts the number of times the logarithm (with base $2$) has to be applied until the result is at most $1$.
			It grows extremely slowly, e.g.\ $\log^* 2^{65536} = 5$.
		}
		\item the initialization for $S = \epsilon$ runs in parallel constant time on $\mathcal{O}(n \log n \log^* n)$ processors,
		\item updates to the data structure after the insertion or deletion of a character are processed in parallel constant time on $\mathcal{O}(n^{\epsilon})$ processors, and
		\item the "LCE" of two given indices is computed in parallel constant time on $\mathcal{O}(n^{\epsilon})$ processors.
	\end{itemize}
\end{theorem}

A maximum size $n$ for the string is fixed at the initialization.

Our algorithm answers "LCE" queries using "string synchronizing sets", which keep track of the occurrences of some substrings within the underlying string.
Kempa and Kociumaka \cite[Section 7]{kempaDynamicSuffixArray2022} show a dynamic algorithm that maintains these sets.
We modify this algorithm to run in parallel constant time, using techniques related to the so-called  Muddling-Lemma \cite{dattaStrategyDynamicPrograms2019}.

As an application of \Cref{the:LCE}, a "common CRCW PRAM" can maintain "Dyck languages" (\Cref{cor:dyck}) and "squares" (\Cref{cor:squares}) in constant time on $\mathcal{O}(n^{\epsilon})$ processors.

\subparagraph{Related Work}
We already discussed related work on "LCE" queries.
Here we go over some related work on dynamic parallel constant-time algorithms.

The study of dynamic parallel constant-time algorithms originated in the field of database theory.
It traces back to the introduction of the complexity class DynFO by Patnaik and Immerman \cite[Section 3.1]{patnaikDynFOParallelDynamic1997}.
DynFO contains dynamic problems where the actions required for processing updates and answering queries can be expressed as first-order formulae.

Patnaik and Immerman show that many relevant problems are in DynFO \cite{patnaikDynFOParallelDynamic1997}, like reachability on undirected graphs.
Their conjecture that reachability on directed graphs is also in DynFO was later proven by Datta et al.\ \cite{dattaReachabilityDynFO2018}.

Schmidt et al.\ study DynFO-algorithms to maintain membership queries for various classes of formal languages \cite{schmidtWorksensitiveDynamicComplexity2021} via "common CRCW PRAMs".
The complexity class DynFO coincides with dynamic parallel constant-time algorithms on "common CRCW PRAMs" \cite[Corollary 5.10]{immerman_1999}.
The advantage of using "PRAMs" instead of first-order formulae to specify dynamic algorithms, as argued by Schmidt et al.\ \cite{schmidtWorksensitiveDynamicComplexity2021}, is that "PRAM" algorithms are usually easier to understand and "PRAMs" have a very clear and natural notion of work.

\subparagraph{Structure}

We first give an informal description of our algorithm in \Cref{sec:informal}.
The detailed description starts with preliminary definitions and auxiliary algorithms in \cref{sec:preliminaries}.
\Cref{sec:data_structure} describes how the algorithm stores its information.
\Cref{sec:queries} explains how to answer queries with this information and \cref{sec:updates} describes how to process updates.
\Cref{sec:applications} shows two applications of \Cref{the:LCE}, namely for maintaining "Dyck languages" and "squares".

\section{Informal Description of the Algorithm}\label{sec:informal}

We first give a rough informal description of the algorithm, split into three parts.
As a foundation we begin with "string synchronizing sets" and how they appear in static algorithms.
Then we move onto a dynamic algorithm for maintaining "string synchronizing sets".
Finally we describe how we modify this dynamic algorithm so that it runs in parallel constant time.

\subparagraph{String Synchronizing Sets in a Static Setting}

A good starting point for "string synchronizing sets" is the algorithm by Kociumaka et al.\ \cite[Section 4]{kociumakaInternalPatternMatching2024}.
We explain its ideas in the following.

Simply speaking, a "string synchronizing set" is a set of occurrences of substrings of a certain size with some useful properties.
An occurrence of a substring is a sequence of positions so that the characters of the underlying string at these positions match the substring.
Occurrences are typically represented by their first position and some identifier for the substring.

A "string synchronizing set" stores occurrences for \emph{some} substrings.
If any occurrence of a substring is stored, then all its occurrences need to be stored.
This property is called \emph{"consistency"}.
Every position of the underlying string is covered by the first half of an occurrence and the total length of all occurrences is linear to the length of the underlying string.
These properties are called \textit{"density"} and \textit{"sparseness"} respectively.
Occurrences are assigned "consistent names" so that occurrences of the same substring have the same name.

Given a hierarchy of logarithmically many "string synchronizing sets" where the occurrence-length doubles with each level, an algorithm can decide efficiently if any two given segments are equal.
To do so, it computes canonical coverings for both segments and checks if they match.
A covering is a selection of occurrences from the "synchronizing sets" so that every position in the segment is covered by at least one occurrence.
To compute a canonical covering, the algorithm selects, from each level in the "synchronizing sets hierarchy", the first and last occurrences that lie entirely within the segment.
See \Cref{fig:query_example} for an example, where $B^1$ is the "string synchronizing set" with substring-length 1, $B^2$ with substring-length 2 and $B^4$ with substring-length 4.
In the string $S$, every first position of an occurrence in $B^1,B^2,B^3$ is marked with the name assigned to that occurrence.

\begin{table}[ht]
	\centering\includegraphics[scale=0.9]{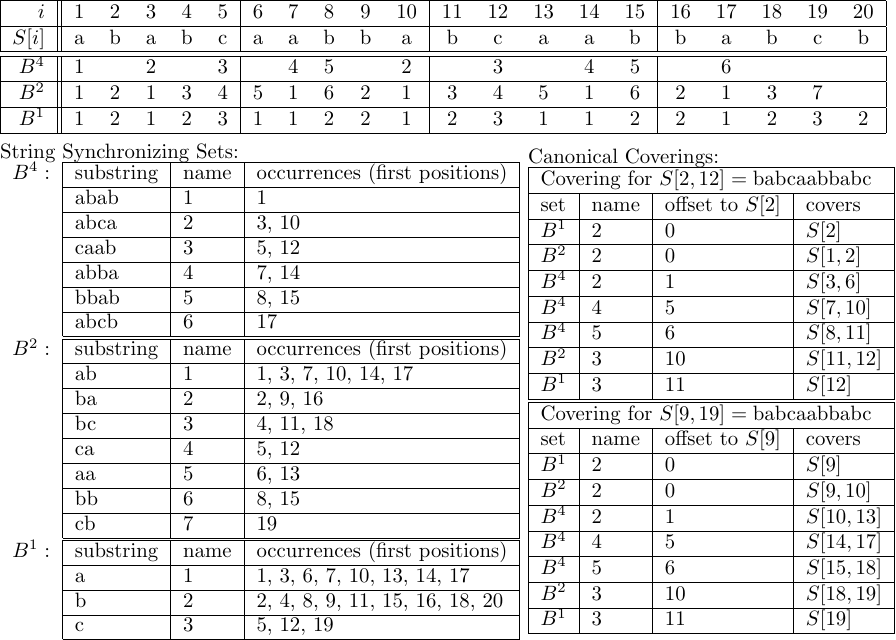}
	\caption{Example for the first three "string synchronizing sets" $B^1, B^2, B^4$ of string $S = \text{ababcaabbabcaabbabcb}$ and canonical coverings for substrings $S[2,12], S[9,19]$. The substrings $S[2,12]$ and $S[9,19]$ are identical and their canonical coverings match.}
	\label{fig:query_example}
\end{table}

It is not hard to see that, thanks to "consistency", two substrings are equal if and only if their canonical coverings match.
The algorithm can compute and compare canonical coverings efficiently because a canonical covering consists of at most a logarithmic number of occurrences.
The "LCE" is then computed with several queries for substring-equality.

The construction algorithm computes the "synchronizing sets hierarchy" via a "hierarchy of decompositions".
Each decomposition divides the underlying string into a sequence of factors.
Unlike the "string synchronizing sets", the factors may have different lengths, but can not overlap one another.
These decompositions are computed in a bottom-up fashion, starting with the decomposition of singleton factors.
Then follows a series of stages where each stage takes the current decomposition and merges factors to produce the next decomposition.

A stage first looks for and merges runs of identical factors.
After this preparation, a stage merges each factor with some of its neighbors consistently, i.e., if the context around two factors is identical, then it produces the same merged factors.
The context around a factor consists of that factor and a certain number of characters before and after that factor.
The exact number of characters depends on the rule used to choose these merges.
We do not look at the concrete rule used by Kociumaka et al.\ \cite[Section 4.2]{kociumakaInternalPatternMatching2024}.
The rule has to be consistent, have a small context, and each merge has to involve up to a constant number of factors.

The algorithm runs over a logarithmic number of stages.
The topmost decomposition consists of only one factor for the entire string.
The algorithm derives the "string synchronizing sets" from the "decompositions hierarchy".
For each level of the "synchronizing sets hierarchy", it selects an appropriate decomposition from the hierarchy.
Every factor in the decomposition is lengthened and becomes an occurrence in the "synchronizing set".

\subparagraph{String Synchronizing Sets in a Dynamic Setting}

Kempa and Kociumaka \cite[Section 7]{kempaDynamicSuffixArray2022} describe a dynamic algorithm that maintains a "string synchronizing sets hierarchy".
The algorithm is sequential and queries and updates run in poly-log-time.
The dynamic algorithm stores the "decompositions hierarchy" and "synchronizing sets hierarchy".
The idea is that, for each update, it suffices to recalculate only small parts of these hierarchies.
To achieve this, Kempa and Kociumaka had to ensure that each update affects only small parts of the hierarchies.
To this end, they substitute the \emph{global} "sparseness" criteria outlined above with \emph{local} "sparseness".
Local "sparseness" demands that the total length of all occurrences within any part of the string is linear in the length of that part, within a $\log^* n$ factor.

To achieve local "sparseness", Kempa and Kociumaka modify the processing of the stages.
They modify the rule for choosing merges so that the merging of each factor depends only on a few factors around it.%
\footnote{
	This is stricter than demanding only "consistency".
	For example, the algorithm in \cite{kociumakaInternalPatternMatching2024} first inspects the entire string to build a rule for merging factors, then applies that rule consistently.
	Changing one character might result in a different rule and lead to many changes in the resulting decomposition.
}
With their modified rule, changing one character of the string affects only a small part of the decompositions and "synchronizing sets".
The dynamic algorithm can afford to recalculate the affected parts for each update.

\subparagraph{String Synchronizing Sets in Parallel Constant-Time}

We take the dynamic algorithm from Kempa and Kociumaka \cite{kempaDynamicSuffixArray2022} and modify it so that updates and queries run in constant time.
Our main challenge is that an update affects logarithmically many levels and it is unclear how the modifications could be computed in parallel constant time.

Our solution is a technique inspired by the so-called Muddling-Lemma \cite{dattaStrategyDynamicPrograms2019}.
The algorithm allows its data structure to become slightly outdated, i.e.~the "synchronizing sets" do not match the current string, but the string as it was a certain number of updates ago.
This technique gives the algorithm enough time to traverse the hierarchy sequentially when calculating the changes that result from an update, but makes queries slightly more complicated.

The processing of one update runs through all $\log n$ levels in sequence, taking $\mathcal{O}(\log^* n)$ time per level.
The algorithm divides this processing into $\log n \log^* n + 1$ actions, each of which runs in constant time.
It spreads the actions across the following updates, so that each following update advances the processing by one action.
The changes that result from the $k$-th update are fully processed with the $(k+\log n \log^* n)$-th update.
Each update computes its first action as well as $\log n \log^* n$ actions belonging to previous updates, all in parallel.

Due to this delayed processing, the lowermost decomposition is delayed by $\log^* n$ updates, the next level by $2\log^* n$ updates, and so on, with the topmost decomposition delayed by $\log n \log^* n$ updates.
These delays carry over to the "synchronizing sets".

Queries have to account for the fact that the information does not reflect the current string.
The solution to this problem lies in the observation that there are only up to $\log^*n \log n$ individual positions with outdated information.
The algorithm cuts queries at these positions, splitting it into a poly-logarithmic number of sub-queries.
It answers each sub-query with the slightly outdated information, then compares the cut-out positions na\"ively.

\section{Preliminaries}
\label{sec:preliminaries}
We first address the model of computation, followed by definitions for strings, "string synchronizing sets", and the dynamic "LCE" problem.
Finally, this section introduces two auxiliary problems used within our "LCE" algorithm.

\subsection{Model of Computation}
We use the ""Common CRCW PRAM"" model as originally described by Shiloach and Vishkin \cite{shiloachFindingMaximumMerging1981}.
A \emph{PRAM} consists of $P(n)$ processors capable of performing the typical elementary operations, where $n$ is the size of the input.
The ticks of the processors are synchronized.

All processors have access to a common memory.
In a \emph{CRCW} PRAM multiple processors can read from and write to cells of this memory concurrently.
In a \emph{common} CRCW PRAM, if multiple processors write to the same cell concurrently, all must write the same value.

We use three measures for the complexity of our algorithms.
These are the runtime, number of processors, and work.
The runtime is the maximum runtime over all processors employed by the "PRAM".
The number of processors is determined at the start of the algorithm and the algorithm is not allowed to create new processors on-the-fly.
The work is the product of the runtime and the number or processors.
When dealing with constant-time algorithms, asymptotic bounds for the work and for the number of processors are interchangeable.
\subsection{Definitions}

\subparagraph{Strings}
A string $S$ is a finite sequence of characters from some alphabet $\Sigma$.
We only consider integer alphabets $\Sigma = [1,n]$.
A decomposition $D = v_1,\ldots,v_m$ of $S$ is a sequence of factors $v_i$ such that $v_1 \circ \ldots \circ v_m = S$.
We write $v_i = S[\alpha,\beta]$ to point to a specific occurrence of $v_i$ so that we can distinguish different occurrences of the same substring.
A string $S$ has a period of length $m$ if $S[i] = S[i+m]$ for all $i \in [1,|S|-m]$.
A string $S$ is periodic if it has a period with length at most $|S|/2$.
A periodic segment of $S$ is some substring $S[\alpha,\beta]$ that is periodic.

\subparagraph{String Synchronizing Sets}

Our algorithm compares substrings via "string synchronizing sets", which were introduced in \cite[Definition 3.1]{kempaStringSynchronizingSets2019}.
We expand the standard definition of "string synchronizing sets" by a function that assigns "consistent names" from some domain $\Delta$ to all occurrences in the set.
Our algorithm uses sets of integers $\{1,2,\ldots,m\}$ for $\Delta$.

\begin{definition}\label{def:synchronizingSet}
	For a string $S$, length $\tau \leq |S|$ and domain $\Delta$, a ""string synchronizing set"" $(B^{\tau},f^{\tau})$ is a set $B^{\tau} \subseteq \{1,\ldots,|S|-\tau\}$ of indices and a function $f^{\tau}:B^{\tau} \to \Delta$ such that:
	\begin{enumerate}
		\item ""Consistency"": If $S[i,i+\tau) = S[j,j+\tau)$ for some $i,j$, then $i \in B^{\tau}$ if and only if $j \in B^{\tau}$.
		\item ""Density"": $B^{\tau} \cap [i,i+\frac{1}{2}\tau) \neq \emptyset$ for all $i \in \{1,\ldots,|S|-\frac{1}{2}\tau\}$, unless there is a periodic segment of size $\frac{3}{2}\tau$ with a period of length at most $\frac{1}{2}\tau$ that contains $i$.
		\item""Sparseness""\footnote{
			We deviate from \cite[Definition 3.1]{kempaStringSynchronizingSets2019} in this property and instead use local sparseness from \cite{kempaDynamicSuffixArray2022}.
		} : $|B^{\tau} \cap [i,i+m)| = \mathcal{O}(\frac{m}{\tau} \log^* \frac{m}{\tau})$ for all $m > 1$ and $1 \leq i \leq n-m-\tau$.
		\item ""Names"" For all $i,j \in B^{\tau}$, $f^{\tau}(i) = f^{\tau}(j)$ if and only if $S[i,i+\tau) = S[j,j+\tau)$.
	\end{enumerate}
\end{definition}

\subparagraph{Dynamic LCE Problem}
Our algorithm solves the dynamic problem \dynamicLCE.

\begin{itemize}
	\item \AP \emph{Problem:} \intro *\dynamicLCE
	\item \emph{Initialization:} String $S=\epsilon$, alphabet $\Sigma = [1,n]$, maximum size $n$
	\item \emph{Updates:}
	\begin{itemize}
		\item $\textsc{Insert}(i,\sigma)$: Set $S$ to $S[1,i-1] \circ \sigma \circ S[i,|S|]$.
		\item $\textsc{Delete}(i)$: Set $S$ to $S[1,i-1] \circ S[i+1,|S|]$.
	\end{itemize}
	\item \emph{Queries:}
	\begin{itemize}
		\item $\textsc{Lcp}(i,j)$: Return $\max\{ m \mid S[i,i+m) = S[j,j+m) \}$.
		\item $\textsc{Lcs}(i,j)$: Return $\max\{ m \mid S(i-m,i] = S(j-m,j] \}$.
	\end{itemize}
\end{itemize}

A dynamic algorithm for \dynamicLCE~is a data structure with sub-algorithms for the initialization, updates, and queries.
The initialization, given the alphabet $\Sigma = [1,n]$ and maximum size $n$, builds a data structure for $S = \epsilon$.
The algorithms for updates modify the data structure according to changes to $S$.
The algorithms for queries use the data structure to compute the length of the "LCE" for two given indices.
We measure the complexity of a dynamic algorithm by the complexity of each sub-algorithm and the size of the data structure.

\subsection{Auxiliary Problems} \label{subsec:auxiliary}

The algorithm relies on two auxiliary dynamic problems.
The first problem \dynamicMarkedString~maintains a dynamic string where some indices are marked.
We use instances of \dynamicMarkedString~to store "string synchronizing sets" and decompositions by marking each first index of an occurrence/factor.
The second problem \dynamicNames~assigns "consistent names" to occurrences and factors in our "synchronizing sets" and decompositions.
We only define the dynamic problems and state the properties of the auxiliary algorithms, but leave descriptions of how they work for the Appendix, \Cref{sec:app:aux}.

\subparagraph{Marked String}

The first algorithm maintains a string with marked indices.
We use it not only to store the current and delayed strings, but also to store "synchronizing sets" and decompositions by marking their first positions.
Updates can insert and delete marks and modify the string.
Marks apply to characters so that, if an update shifts a character, a mark on that character shifts with it.
There are queries for retrieving substrings and for determining the nearest marked neighbor of a given index.

\begin{itemize}
	\item \AP \emph{Problem:} \intro *\dynamicMarkedString
	\item \emph{Initialization:} Alphabet $\Sigma = [1,n]$, string $S = \epsilon$, marked indices $Z = \emptyset$, maximum size $n$
	\item \emph{Updates:}
	\begin{itemize}
		\item $\textsc{Insert}(i,\sigma): \begin{cases}
			\text{Set $S$ to $S[1,i-1] \circ \sigma \circ S[i,|S|]$.}\\
			\text{Set $Z$ to $\{ j \mid j \in Z, j < i \} \cup \{ j+1 \mid j \in Z, j \geq i \}$.}
		\end{cases}$
		\item $\textsc{Delete}(i): \begin{cases}
			\text{Set $S$ to $S[1,i-1] \circ S[i+1,|S|]$.}\\
			\text{Set $Z$ to $\{ j \mid j \in Z, j < i \} \cup \{ j-1 \mid j \in Z, j > i \}$.}
		\end{cases}$
		\item $\textsc{Mark}(X)$ (with $X \subseteq [1,n]$): Set $Z$ to $Z \cup X$.
		\item $\textsc{Unmark}(X)$ (with $X \subseteq [1,n]$): Set $Z$ to $Z \setminus X$.
	\end{itemize}
	\item \emph{Queries:}
	\begin{itemize}
		\item $\textsc{SubStr}(i,m)$: Return $S[i,i+m)$.
		\item $\textsc{Succ}(i)$: Return $\min\{ j \in Z \mid j > i \}$, or $\bot$ if none exists.
		\item $\textsc{Pred}(i)$: Return $\max\{ j \in Z \mid j < i \}$, or $\bot$ if none exists.
	\end{itemize}
\end{itemize}

The problem \dynamicMarkedString~has batch-updates, which mark or unmark several indicies at the same time.
The size of the batches in our use-cases is $|X| = \mathcal{O}(\log^* n)$.

\begin{lemma}\label{lem:shiftingSuccessor}
	For any $\epsilon > 0$, there is a dynamic algorithm for \dynamicMarkedString~on a "common CRCW PRAM", where
	\begin{itemize}
		\item the size of the data structure is in $\mathcal{O}(n)$,
		\item the initialization runs in parallel constant time on $\mathcal{O}(n^{\epsilon})$ processors,
		\item \bettersc{Insert} and \bettersc{Delete} run in parallel constant time on $\mathcal{O}(n^{\epsilon})$ processors,
		\item \bettersc{Mark} and \bettersc{Unmark} run in parallel constant time on $\mathcal{O}(n^{\epsilon} |X|^{2+\epsilon})$ processors,
		\item \bettersc{SubStr} runs in parallel constant time on $\mathcal{O}(mn^{\epsilon})$ processors, and
		\item \bettersc{Succ} and \bettersc{Pred} run in parallel constant time on $\mathcal{O}(n^{\epsilon})$ processors.
	\end{itemize}
\end{lemma}

\subparagraph{Names}

An important task within the algorithm is assigning short names to bigger objects.
These objects are substrings from $\Sigma^m$ that represent factors/occurrences.
The names are from a domain $\Delta$.
The primary purpose of this auxiliary algorithm is to maintain a function $g$ that maps objects $x$ to names $g(x) \in \Delta$, without naming-conflicts and with $|\Delta| < |\Sigma|^m$.

This auxiliary algorithm avoids naming-conflicts so long as there are never more than $|\Delta|/k$ distinct objects that need names at the same time, called \emph{tracked objects}.
The number $k > 1$ describes how many objects an update can modify at once.
The tracked objects are stored in a set $Z$.
The auxiliary algorithm maintains a function $f$ that maps objects $x$ to $f(x) \in [1,n]$ such that, if $f(x) \leq 0$, it can \emph{forget} $x$, removing it from the tracked objects.

\AP Updates receive a multi-set $X$ of objects.
Let $\intro *\setcount {X}{x}$ be the number of occurrences of $x$ in $X$.
Updates $\textsc{Add}(X)$ first expand the tracked objects and assign new names for previously untracked objects in $X$.
Then they increase $f(x)$ for each unique $x \in X$ by $\setcount{X}{x}$.
Updates $\textsc{Sub}(X)$ decrease $f(x)$ for each unique $x \in X$ by $\setcount{X}{x}$, then forget every $x$ with $f(x) \leq 0$.

Queries retrieve the name $g(x)$ of a given object $x$, if it is a tracked object.

\begin{itemize}
	\item \AP \emph{Problem:} \intro *\dynamicNames
	\item \emph{Initialization:} Alphabet $\Sigma$, domain $\Delta$, set $Z = \emptyset$, function $f: Z \to [1,n]$, injective function $g: Z \to \Delta$, maximum size $n$ of $Z$, substring size $m$.
	\item \emph{Updates:}
	\begin{itemize}
		\item $\textsc{Add}(X):\begin{cases}
			\text{For each unique $x \in X$, if $x \in Z$, set $f(x)$ to $f(x)+\setcount{X}{x}$.}\\
			\text{If $x \notin Z$ add $x$ to $Z$ with $f(x) = \setcount{X}{x}, g(x) \in \Delta$ so that $g$ is injective.}
		\end{cases}$
		\item $\textsc{Sub(X)}:\begin{cases}
			\text{For each unique $x \in X$, if $x \in Z$, set $f(x)$ to $f(x) - \setcount{X}{x}$.}\\
			\text{Then, for each $x \in X$, if $f(x) \leq 0$, remove $x$ from $Z$.}
		\end{cases}$
	\end{itemize}
	\item \emph{Queries:}
	\begin{itemize}
		\item $\textsc{Name}(x)$: Return $g(x)$, or $\bot$ if $x \notin Z$.
	\end{itemize}
\end{itemize}

In our usage, the batches are very small with $|X| \leq k = \mathcal{O}(\log^* n)$.
The names are from the domain $\Delta = [1,nk]$, which is dependent on the maximum batch size.

\begin{lemma}\label{lem:names}
	For any $\epsilon > 0$, there is a dynamic algorithm for \dynamicNames~on a "common CRCW PRAM", where
	\begin{itemize}
		\item the size of batches is $|X| \leq k$ and the domain is $\Delta = [1,nk]$,
		\item the size of the data structure is in $\mathcal{O}(n)$,
		\item the initialization in parallel constant time on $\mathcal{O}(nk)$ processors,
		\item updates run in parallel constant time on $\mathcal{O}(2^kmn^{\epsilon})$ processors, and
		\item queries run in parallel constant time on $\mathcal{O}(mn^{\epsilon})$ processors.
	\end{itemize}
\end{lemma}

This concludes the preliminary definitions.
Next we describe our "LCE" algorithm.
We first define the data structure in \cref{sec:data_structure}.
Then we describe how to query it in \cref{sec:queries}.
Finally we explain how to update the data structure in \cref{sec:updates}.

\section{Data Structure}
\label{sec:data_structure}
We describe how the algorithm stores its information.
It uses several instances of the auxiliary problems described in \Cref{subsec:auxiliary}.
\Cref{tab:data} gives a brief overview of the data structures we will define in this section for easy reference.

\begin{table}[ht]
	\begin{center}
		\begin{tabular}{|c|c|l|}
			\hline
			$\calS $&\dynamicMarkedString&current string, marks unused\\
			$\calR $&compact sorted array&record of recent changes\\ \hline
			
			\multicolumn{3}{l}{for $\tau = 2^z, z \in \mathbb{N}, \tau \in [1,n]$:}\\ \hline
			$\calBtau $&\dynamicMarkedString&delayed string, marks occurrences in $(B^{\tau},f^{\tau})$\\
			$\calNBtau $&\dynamicNames&names used by $(B^{\tau},f^{\tau})$\\ \hline
			
			\multicolumn{3}{l}{for $\ell \in \{1,\log n \}$:}\\ \hline
			$\calSl $&\dynamicMarkedString&delayed string, marks factors in $\Dell $\\
			$\calDl $&\dynamicMarkedString&sequence of factors names in $\Dell $\\
			$\calNDl $&\dynamicNames&names used by $\Dell $\\ \hline
		\end{tabular}
	\end{center}
	\caption{Data Structures used by the algorithm. Only $\calS $ and $\calR $ are kept fully up-to-date.}
	\label{tab:data}
\end{table}

\subparagraph{Up-to-date Information}

The first data structure stores the current string without delay in an instance $\intro *\calS $ of \dynamicMarkedString.
It does not use the marks.
Most of the other data structures represent the string as it was some number of updates ago.

The algorithm uses a record of the most recent updates to resolve outdated information.
The record $\intro *\calR $ is a compact sorted array that stores the most recent $\log n \log^* n$ updates.
The entries in $\calR $ store either that a character was inserted at some index or that at least one character was deleted between two indices.
The entries are sorted by index.
Each entry also stores the age of that change.

\subparagraph{String Synchronizing Sets}
The algorithms uses a ""synchronizing sets hierarchy"" where the length of the substrings doubles with each level.

For each $\tau = 2^z, z \in \mathbb{N}, \tau \in [1,n]$, it stores a "synchronizing set" $(B^{\tau}, f^{\tau})$ in an instance $\intro *\calBtau $ of \dynamicMarkedString, with some delay.
The underlying string in $\calBtau $ is the delayed string at that level.
It marks every index $i \in B^{\tau}$ and annotates them with the name $f^{\tau}(i)$.
It maintains forward- and backward-links for every index in $\calBtau $ to the corresponding index in the current string $\calS $.
It supplements the "synchronizing sets" with an instance $\intro * \calNBtau $ of \dynamicNames~to maintains the "names" used by $B^{\tau}$.

\subparagraph{Decompositions}
The algorithm computes the "synchronizing sets" via a ""decompositions hierarchy"" of decompositions $\Dell $ for each $\ell \in \{1,\ldots, \log n\}$.

The algorithm maintains several properties for the decompositions $\intro *\Dell $ that mirror the properties of "string synchronizing sets".
The decompositions are consistent in the sense that, if a big enough substring around a factor also occurs somewhere else, there is a copy of that factor at that second occurrence as well.
It furthermore maintains upper and lower bounds for the size of the factors.
We prove these properties at a later point.

\begin{lemma}\label{lem:factorSize}
	Let $u$ be a factor in the decomposition $\Dell $ at level $\ell$.
	\begin{enumerate}
		\item If there are no string $w$ and $z \geq 2$ with $u = w^z$, then $|u| \leq 2^{\ell+1}$.
		\item If $u$ and $v$ are adjacent factors in $\Dell $, then $|u| + |v| \geq 2^{\ell+1}$.
	\end{enumerate}
\end{lemma}

The instance $\intro *\calSl $ of \dynamicMarkedString~stores the delayed string.
It marks every first index of a factor in $\Dell $ and annotates them with the name of the factor.
It maintains links between each index in $\calSl $ and the corresponding index in the current string $\calS $.
	
The instance $\intro *\calDl $ of \dynamicMarkedString~also stores the decomposition.
Unlike $\calSl $, the underlying string is \emph{not} the string for \dynamicLCE~but the sequence of factor names in the decomposition $\Dell $.
It marks every factor that has a different name than the factor before it.
It maintains links between each factor name in $\calDl $ and its starting index in $\calSl $.

It uses the instance $\intro *\calNDl$ of \dynamicNames~to maintain consistent factor names for $\Dell $.

\section{Queries}
\label{sec:queries}
We first show how to answer substring equality queries using a logarithmic "synchronizing sets hierarchy".
Then we explain how to handle delayed "synchronizing sets".
We conclude by showing that "LCE" queries can be answered with substring equality queries.

\subparagraph{Substring Equality}
\label{subsec:queries:substring}
We first show that we can answer substring equality queries if fully up-to-date "string synchronizing sets" $(B^{\tau},f^{\tau})$ for each $\tau = 2^z, z \in \mathbb{N}, \tau \in [1,n]$ are available.
For now, we assume that we also have, for each $(B^{\tau},f^{\tau})$, successor- and predecessor-links which point, for each index of $S$, to the closest occurrence in $B^{\tau}$.
\begin{lemma}\label{lem:queryNoMuddling}
	There is a parallel constant-time algorithm that, given a string $S$ of size $n$ and integers $i,j,m$, decides if $S[i,i+m) = S[j,j+m)$ holds on a "common CRCW PRAM" with $\mathcal{O}(\log n)$ processors.
	For each $\tau = 2^z, z \in \mathbb{N}, \tau \in [1,n]$, the algorithm has access to a "string synchronizing set" $(B^{\tau},f^{\tau})$ and successor- and predecessor-links for $B^{\tau}$.
\end{lemma}
\emph{Sketch.}
The algorithm first covers $S[i,i+m)$ with logarithmically many occurrences from the "synchronizing sets hierarchy".
It uses the successor- and predecessor-links to find, for each level $\tau$, the first and last occurrence at that level that lies entirely in $S[i,i+m)$.
This produces a covering of $S[i,i+m)$, unless there are periodic prefixes or suffixes, which require some special handling.
Then the algorithm looks at $S[j,j+m)$ and, if it can find the same covering there, concludes $S[i,i+m) = S[j,j+m)$.
See the Appendix, \Cref{sec:queryNoMuddling} for a proof.

\subparagraph{Accounting for Delay}
\label{subsec:queries:muddling}

As mentioned above, the algorithm has to deal with the complication that it does not maintain fully up-to-date "string synchronizing sets".
Indeed there may be $\log n \log^* n$ positions where characters were recently inserted or deleted for which its information is outdated.
Our solution is to cut the query at these positions into sub-queries which can be answered with the available information.

\begin{lemma}\label{lem:queryMuddling}
	For any $\epsilon > 0$, there is a parallel constant-time algorithm that, given the information described in \Cref{sec:data_structure} and integers $i,j,m$, decides if $S[i,i+m) = S[j,j+m)$ holds on a "common CRCW PRAM" with $\mathcal{O}(n^{\epsilon})$ processors.
\end{lemma}
\begin{proof}
	The algorithm first splits the query into sub-queries so that no sub-query contains recently changed indices.
	It creates a sub-query for the prefixes of $S[i,i+m)$, $S[j,j+m)$ up to but not including the first recorded change in either segment.
	Then, for every entry in the record $\calR $, if the recorded change lies within $S[i,i+m)$, it creates a further sub-query that starts at the first unaffected index after that change and ends at the next change in either $S[i,i+m)$ or $S[j,j+m)$.
	It similarly creates sub-queries for changes in $S[j,j+m)$.
	
	This results in $\mathcal{O}(\calR ) = \mathcal{O}(\log n \log^* n)$ sub-queries, which are answered in parallel using \Cref{lem:queryNoMuddling}.
	Since the structures $\calBtau $ are delayed, it is possible that the same character occupies different indices in $\calS $ and $\calBtau $, even if this position itself was not changed recently.
	The algorithm uses the links between each position in the current string $\calS $ and their corresponding positions in the delayed structures $\calBtau $ to match shifted indices.
	
	The complexity of each sub-query is given in \Cref{lem:queryNoMuddling}, but there is one caveat.
	In \Cref{lem:queryNoMuddling}, we assumed that successor- and predecessor-links are precomputed and accessible with constant work, which is not true here.
	Instead the algorithm retrieves them from $\calBtau $, which requires $\mathcal{O}(n^{\theta})$ processors per query for any $\theta > 0$ (\Cref{lem:shiftingSuccessor}).
	For any given $\epsilon > 0$, we can choose $\theta$ such that the total number of processors overall is $\mathcal{O}(n^{\epsilon})$.
	
	Finally, alongside the sub-queries for unchanged segments, the algorithm checks single indices which were recently inserted and lie within $S[i,i+m)$ or $S[j,j+m)$ and compares them na\"ively.
	There are at most $\log n \log^* n$ recently inserted indices.
\end{proof}

\subparagraph{LCE Queries}
\label{subsec:queries:LCE}
The final step reduces "LCE" queries to substring equality queries.

\begin{lemma}\label{lem:LCE}
	For any $\epsilon > 0$, there is a parallel constant-time algorithm that, given the information described in \Cref{sec:data_structure}, answers "LCE" queries on a "common CRCW PRAM" with $\mathcal{O}(n^{\epsilon})$ processors.
\end{lemma}
\begin{proof}
	The algorithm finds the length of the "LCE" via several substring equality queries with a binary-search-like approach, with arity $n^{\epsilon}$ instead of two.
	Starting with the search range of all lengths $[1,n]$, it splits the search range into $n^{\epsilon}$ same-sized segments at each step.
	It formulates one substring equality query per segment to determine which segment contains the "LCE" and becomes the next search range.
	This procedure terminates after $\frac{1}{\epsilon}$ rounds.
\end{proof}
The information described in \Cref{sec:data_structure} suffices to answer "LCE"-queries.
It remains to show that the algorithm can maintain that information in parallel constant time on $\mathcal{O}(n^{\epsilon})$ processors.

\section{Updates}
\label{sec:updates}
We first explain how the algorithm spreads the processing of one change across the following updates.
Then we take a look at the actions necessary to process an update.

\subparagraph{Delayed Processing}
The decompositions $\Dell $ are arranged bottom-up along levels $\ell \in \{1,\ldots,\log n\}$.
Updates traverse these levels in ascending order.
The processing of an update is split into $\log n\log^* n + 1$ actions $A(k,\ell,i)$, where $k$ counts the updates, $\ell  \leq \lceil \log n \rceil$ the levels, and $i \leq \lceil \log^* n \rceil$ the actions on a level.
The actions of the $k$-th update are $\{ A(k,0,0) \} \cup \{ A(k,\ell, i) \mid \ell \in [1,\log n], i \in [i, \log^* n] \}$.
Actions are ordered lexicographically, i.e.~$A(k,\ell,i) < A(k,\ell',i')$ if $\ell < \ell'$ or $(\ell = \ell') \land (i < i')$.
All of these actions processed sequentially in ascending order compute the changes that result from the $k$-th update.

The first action $A(k,0,0)$ updates the current string $\calS $ and the record $\calR $.
The actions $A(k,\ell,1),\ldots,A(k,\ell,\log^* n)$ together calculate the changes to level $\ell$ of the decompositions hierarchy. They also update the "synchronizing sets" that belong to that level.

The algorithm only processes the very first action $A(k,0,0)$ immediately.
Then, for each subsequent update, it advances the processing of this update by one further action.
When update $k$ arrives, the algorithm thus processes the following actions in parallel (if they exist):
\[
\begin{array}{rlclcll}
	A(&k&,&0&,&0&),\\
	A(&k-1&,&1&,&1&),\\
	\vdots&&&&&&\\
	A(&k-\log^* n&,&1&,&\log^* n&),\\
	\vdots&&&&&&\\
	A(&k - \log n\log^* n&,& \log n&,& \log^* n&).
\end{array}
\]

\subparagraph{Affected and Relevant Segment}
For a decomposition $\Dell $ and an update that affects index $i$, the "affected segment" is the sequence of factors $v_j,\ldots,v_{j+k}$ in $\Dell $ that might be affected by that update.
It is the part of the $\Dell $ that has to be recalculate for an update at index $i$.
An important prerequisite of our algorithm is that this "affected segment" is small.

The "relevant segment" is the union of all factors that influence a factor in $v_j,\ldots,v_{j+k}$.
They are the factors the algorithms has to know to recalculate the "affected segment".
\begin{lemma}\label{lem:segments}
	Let $\Dell  = v_1,\ldots,v_m$ be the decomposition at level $\ell$, let $i$ be some index and let $v_j$ be the factor that contains the $i$-th character.
	The ""affected segment"" for $i$ in $\Dell $ is $v_{j-\log^* n-9},\ldots,v_{j+5}$.
	The ""relevant segment"" is $v_{j-\log^* n - 14},\ldots,v_{j+\log^* n + 14}$.
\end{lemma}
This lemma establishes the context that affects each factor of the decomposition.
It informs us that changes to one factor affect at most $\log^* n +9$ factors before and $5$ factors after it.
To recalculate these affected factors, the algorithm only needs $\log^* n + 14$ factors on both sides.
We do not prove \Cref{lem:segments} yet because we have to present the algorithm first.

\subsection*{Processing a Level}
We describe how the actions $A(k,\ell,1),\ldots,A(k,\ell,\log^*n)$ compute the changes to one level $\ell$ in the form of one parallel algorithm with runtime $\mathcal{O}(\log^* n)$.
The processing of this algorithm is then split across $\log^* n$ actions where the final action $A(k,\ell, \log^* n)$ applies changes to $\Dell $.

The changes the final action $A(k,\ell,\log^*n)$ applies to the decomposition $\Dell $ affect later actions for both this update $A(k,\ell+1,1)$ and later updates $A(k',\ell+1,1), k' > k$.

\subparagraph{Preparation}
The algorithm modifies the "affected segment" of the decomposition $D^{\ell-1}$ to ensure two properties.
The preparation enforces that factors shorter than $2^{\ell-1}$ are surrounded by factors longer than $2^{\ell}$, which is required for the lower bound of the size-constraints.
It also enforces that neighboring factors are distinct, which is required by the next step.
It does this by merging factors, using the marks in $\mathcal{D}^{\ell-1}$ to identify runs of identical factors.

The result of the preparation is a string $d_1d_2\ldots$ of the names of the factors in the "relevant segment" of $D^{\ell-1}$.
To achieve the upper bound of the size-constraints, it "deactivates factors" longer than $2^{\ell}$.
A ""deactivated factor"" will not be merged in the next step.
It "deactivates a factor" by replacing the name of that factor with $-1$.
The final string satisfies the following.
\begin{lemma}\label{lem:prep}
	Let $d_1d_2\ldots$ be the string of factor names produced by level $\ell$ preparation and let $v_1,v_2,\ldots$ be the corresponding sequence of factors.
	The following hold for every $i, j$:
	\begin{enumerate}
		\item $d_i = -1$ if and only if $|v_i| > 2^{\ell}$.
		\item If $d_i \neq -1$, then $v_i = v_j$ if and only if $d_i = d_j$.
		\item If $|v_i| < 2^{\ell-1}$, then $d_{i-1} = d_{i+1} = -1$.
		\item If $d_i \neq -1$, then $d_i \neq d_{i+1}$.
	\end{enumerate}
\end{lemma}

See the Appendix, \Cref{sec:preparation} for a more detailed description and a proof of \Cref{lem:prep}.

\subparagraph{Choosing Merges}
This step starts with the string of factor names from the preparation.
Recall that the algorithm maintains names according to \Cref{lem:names}, which uses integers in  $[1,N]$ with $N = \mathcal{O}(n\log^* n)$.
It maps these names to $[0,N-1]$ which, combined with $-1$ for "deactivated factors", gives a string over $[-1,N-1]$.
The algorithm applies a function to this string that yields a bit-string where the beginning of each merged factor is $1$.

There are several properties this bit string has to have.
The bit string may not contain $0000$, because this would result in the merging of five or more factors, which could violate the upper bound for the size of merged factors.
Likewise the bit string may not contain consecutive $1$'s, because this would result in a factor that is not merged, which could violate the lower bound.
The only exception to this rule is that "deactivated factors" should not be merged, meaning that the bit string must have consecutive $1$'s there.
To achieve a small context, each bit may depend on only a small sub-sequence of factor names.
And finally, it must be possible to compute the function that produces this bit string efficiently.

We use deterministic coin flipping from \cite[Lemma 1]{mehlhornMaintainingDynamicSequences1997}, with some modifications, to construct this function.
This approach was taken from Kempa and Kociumaka, but with different modifications and therefore different properties compared to \cite[Theorem 7.3]{kempaDynamicSuffixArray2022}.

\begin{lemma}\label{lem:coinFlipping}\cite[Theorem 7.3]{kempaDynamicSuffixArray2022} \cite[Lemma 1]{mehlhornMaintainingDynamicSequences1997}
	For every $N, n$ there is a function $f:[-1,N-1]^n \to [0,1]^n$ which satisfies the following for all strings $S \in [-1,N-1]^n$ and indices $i,j \in [1,n]$:
	\begin{enumerate}
		\item\label{enum:coinFlip:localdist} If $j \notin [i-4,i+\log^* N + 5]$, then changing $S[j]$ does not affect $f(S)[i]$.
		\item\label{enum:coinFlip:localforce} If there is a $k$ with $S[k] = -1$, $\min\{i,j\} \lneq k \lneq \max\{i,j\}$, then changing $S[j]$ does not affect $f(S)[i]$.
		\item\label{enum:coinFlip:lowerforce} If $S[i] = -1$, then $f(S)[i,i+1] = 11$.
		\item\label{enum:coinFlip:lower} If $f(S)[i] = 1$ and $S[i] \neq -1$, then $f(S)[i+1] = 0$.
		\item\label{enum:coinFlip:upper} $f(S)$ does not contain $0000$.
		\item\label{enum:coinFlip:work} There is a parallel algorithm that calculates $f(S)$ in $\mathcal{O}(\log^* N)$ time on a "common CRCW PRAM" with $n$ processors.
	\end{enumerate}
\end{lemma}

We give the proof of \Cref{lem:coinFlipping} in the Appendix, \Cref{sec:coinFlipping}.
It follows from \Cref{lem:coinFlipping} that all factors are merged into groups of two or four unless they are "deactivated" or surrounded by "deactivated factors".
\Cref{lem:factorSize} follows by induction using this fact and the preparation.

\Cref{lem:segments} follows from the combination of property \ref{enum:coinFlip:localdist} of \Cref{lem:coinFlipping} and the preparation.
See the Appendix, \Cref{sec:proof:lem:segments} for a proof.

\subparagraph{Finalization}
The finalization has two tasks. It updates the decomposition $\Dell $, then updates the "synchronizing sets" that are derived from $\Dell $.

For the first task, it updates $\calSl$ and $\calNDl$ according to how the merges were chosen.

The instance $\calNDl$ uses as its keys the covering of a factor with the "synchronizing sets hierarchy", as would be done to answer a query (see \Cref{lem:queryNoMuddling}).
This might produce a faulty covering because the queries in \Cref{lem:queryNoMuddling} do not account for delayed information.
We are safe to ignore this situation because it only occurs if there is a younger update with overlapping "affected segment".
When that younger update reaches this level, it will recalculate the factor with the faulty covering, thus cleaning up any potential problems.

After updating the decomposition in $\calSl $ and assigning names as needed, the algorithm updates $\calDl $.
It can insert the new merged factor names into $\calDl $ in a compact manner since this involves only $\mathcal{O}(\log^* n)$ factor names in total.

Both the finalization of level $\ell$ and the preparation of level $\ell+1$ modify the decomposition $\Dell $ and access $\calSl $, $\calDl $ and $\calNDl $.
To avoid conflicts, the finalization of $\ell$ waits until the preparation of $\ell+1$ has finished before it accesses $\calSl $, $\calDl $ and $\calNDl $.

The second task is deriving the "synchronizing sets" from the decompositions.
The set $B^{\tau}$ contains an element $i$ for each factor $v$ in decomposition $\Dell $.
To ensure "consistency" the occurrence $S[i,i+\tau)$ has to contains the entire context of $v$.
We first show how to match "synchronizing sets" to decompositions to guarantee that this is the case.
We combine \Cref{lem:factorSize} and \Cref{lem:segments} to derive an upper bound for the context.
It states that, if a large enough context around two factors is identical, then the factors are identical too.

\begin{corollary}\label{lem:contextSize}
	For $\ell > 0$, let $\alpha_{\ell} = 5(2^{\ell+1}-1)$ and $\beta_{\ell} = (\log^*N+9)(2^{\ell+1}-1)$.\\
	Let $D_{\ell},D'_{\ell}$ be the level-$\ell$ decompositions of strings $S,S'$ respectively.
	If $x = S[i,i+m)$ is a factor in $D_{\ell}$ and there is a $j$ with $S[i-\alpha_{\ell},i+\beta_{\ell}) = S'[j-\alpha_{\ell},j+\beta_{\ell}]$, then $x = S'[j,j+m)$ is a factor in $D'_{\ell}$.
\end{corollary}

Each factor in $\Dell $ becomes an occurrence by shifting the first index $\alpha_{\ell}$ positions to the left.
It follows from \Cref{lem:contextSize} that this produces a consistent "string synchronizing set" if $\tau$ is at least as big as the context $\alpha_{\ell} + 1 + \beta_{\ell}$.

\begin{lemma}\label{lem:syncSet}
	Let $D_{\ell}$ be the level-$\ell$ decomposition of a string $S$, let $z$ be some integer and $\tau = 2^z$.
	If $\alpha_{\ell}+1+\beta_{\ell} \leq \tau < \alpha_{\ell+1}+1+\beta_{\ell+1}$, then $B_{\tau} = \{ i - \alpha_{\ell} \mid \exists j. S[i,j) \text{ is a factor in } D_{\ell} \}$, alongside an appropriate name-function $f_{\tau}$, is a "string synchronizing set" for $S$.
\end{lemma}

Level $\ell$ updates "synchronizing sets" according to \Cref{lem:syncSet}.
"Consistency" follows from the fact that the occurrence contains the context.
"Density" follows from the upper bound of \Cref{lem:factorSize} and "sparseness" follows from the lower bound.

\section{Applications}
\label{sec:applications}
We show two applications of our algorithm, one from the area of formal languages and another from string algorithms.

We note that these algorithms support only character substitutions, so inserting or deleting characters is not possible.
The first algorithms also allow "void characters", which can be substituted like any other character, but are ignored for the purpose of the problem.

\subsection{Dyck Languages}
The ""Dyck language"" $D_k$ is the language of well-balanced parentheses expressions with $k$ types of parentheses.
We first note that we can improve the bound for the number of processors of a dynamic constant-time algorithm for $D_1$ from $\mathcal{O}((\log n)^3)$ in Schmidt et al.\ \cite[Theorem 6.5]{schmidtWorksensitiveDynamicComplexity2021} to $\mathcal{O}(\log n \log \log n)$.
We prove this in the Appendix, \Cref{sec:dyck1}.
\begin{theorem}\label{the:dyck1}
	There is a dynamic parallel constant-time algorithm that maintains membership in $D_1$ on a "common CRCW PRAM" with $\mathcal{O}(\log n \log \log n)$ processors.
	Updates allow only character substitutions and "void characters".
\end{theorem}
Schmidt et al.\ also consider "Dyck languages" $D_k$ for $k > 1$.
They prove a link between the complexity for "Dyck languages" and the problem \stringEquality~\cite[Lemma 6.9]{schmidtWorksensitiveDynamicComplexity2021}.
The problem \stringEquality, which maintains if two dynamic strings are equal, is reducible to \dynamicLCE.
They show that, given an efficient algorithm for \stringEquality, it is possible to maintain membership in $D_k$ while increasing the work by only a logarithmic factor.
Applying \Cref{the:LCE} to \cite[Lemma 6.9]{schmidtWorksensitiveDynamicComplexity2021} yields the following bound.
\begin{corollary}\label{cor:dyck}
	For any $\epsilon > 0, k \in \mathbb{N}$, there is a dynamic parallel constant-time algorithm that maintains membership in $D_k$ on a "common CRCW PRAM" with $\mathcal{O}(n^{\epsilon})$ processors.
	Updates allow only character substitutions and "void characters".
\end{corollary}
See the Appendix, \Cref{sec:proof:dyck}, for a proof.

\subsection{Squares}

A ""square"" is a substring $S[\alpha,\beta] = vv$ where $v$ is not empty.
A string is ""square-free"" if it contains no "squares".
Amir et al.\ show dynamic sequential algorithms that run in $n^{o(1)}$ time for "square-freeness" \cite[Theorem 21]{amirRepetitionDetectionDynamic2019} and maintaining all "squares" \cite[Theorem 26]{amirRepetitionDetectionDynamic2019}.
With some modifications and using our "LCE" algorithm, we can construct similar efficient algorithms that run in parallel constant time.

\begin{corollary}\label{cor:squares}
	For any $\epsilon > 0$, there is a dynamic parallel constant-time algorithm that answers if a dynamic string is "square-free" and maintains all "squares" on a "common CRCW PRAM" with $\mathcal{O}(n^{\epsilon})$ processors.
	Updates to the string allow only character substitutions.
\end{corollary}

We describe the modified algorithm for "square-freeness" in the Appendix, \Cref{sec:proof:squares}.
Roughly speaking, we start with the algorithm by Amir et al., then carefully go through it to identify processes that are independent of each other and can run in parallel.
This leaves only two parts that require additional insight.
The first straightforward part is that the algorithm uses \dynamicLCE~to compare long substrings, which we solve with \Cref{the:LCE}.
The second part is that the algorithm uses a problem called \dynamicPrefixSuffix, which it solves recursively with logarithmic depth.
We modify the recursion so that, instead of splitting the current string into a constant number of fragments at each step, it splits it into $n^{\epsilon}$ fragments, thus reducing the depth of the recursion to constant.

\section{Conclusion}
\label{sec:conclusion}
We gave a dynamic parallel constant-time algorithm for "LCE" queries on $\mathcal{O}(n^{\epsilon})$ processors, for any $\epsilon > 0$.
It combines "string synchronizing sets", a technique from string algorithms, with a technique related to the Muddling-Lemma, from dynamic constant-time algorithms.
We also looked at applications of our algorithm.
Using our algorithm, we can maintain "Dyck languages" and "squares" in parallel constant time on $\mathcal{O}(n^{\epsilon})$ processors.

We list some possible avenues of further research.

"LCE" queries have many applications and we expect there to be other applications of our algorithm.
"String synchronizing sets" likewise have been used for many different purposes, like internal pattern matching or suffix arrays.
This raises the question if it is possible to modify such algorithms to run in parallel constant-time, even with delayed information.

Our algorithm has two limitations.
A maximum string size has to be set during the initialization and the algorithm supports only integer alphabets.
While it is possible to overcome the size limitation, as we show in the Appendix, \Cref{sec:SizeLimit}, finding algorithms for other alphabets remains an open question.

And finally, we would ideally like to have a \emph{provably} work-optimal constant-time algorithm for dynamic "LCE" queries.
This requires an algorithm with matching lower bound.
Unfortunately non-trivial lower bounds that apply specifically to dynamic parallel constant-time algorithms have so far been elusive, even for other problems.




\bibliography{bib}


\appendix

\section{Auxiliary Problems}\label{sec:app:aux}

We describe algorithms for the auxiliary problems mentioned in \Cref{subsec:auxiliary}.
We build these algorithms up from other auxiliary problems.

\subsection{Static Algorithms}
We first recite three known static algorithms which we will employ in our dynamic algorithms.

\subparagraph{First Occurrence}
The first algorithm finds the first occurrence of a given element in an array with linear work.
This problem is trivial for sequential algorithms.
In fact, the na\"ive algorithm that visits all indices in sequence is clearly optimal.
It is however more challenging to solve in parallel constant time.
The na\"ive constant-time algorithm checks, for all indices in parallel, if that index is the first occurrence.
An index is the first occurrence if it is an occurrence and no previous index is an occurrence.
This takes quadratic work.
Fich et al.\ improve this with an algorithm that takes only linear work.

\begin{lemma}\label{lem:linSearch}\cite[Theorem 1]{fichRelationsConcurrentwriteModels1984}
	There is a parallel constant-time algorithm that, given an array $X$ of size $n$ and element $e$, finds the minimum $i$ with $X[i] = e$ on a "common CRCW PRAM" with $\mathcal{O}(n)$ processors.
\end{lemma}

\begin{proof}
	The algorithm divides $X$ into $\sqrt{n}$ equal-sized segments and checks for each segment if it contains $e$ anywhere within the segment.
	Then it finds the first segment that contains $e$ na\"ively.
	Finally it searches that segment for the first $e$ na\"ively.
	Each na\"ive search runs on $\mathcal{O}(\sqrt{n}^2) = \mathcal{O}(n)$ processors, because both times, the input has size $\mathcal{O}(\sqrt{n})$.
\end{proof}

\subparagraph{Maximum}
The second search algorithm finds the maximum with almost-linear work.
For this problem, it is of note that we use very \emph{general} elements, where we can only determine if one element is bigger than another.
We do not restrict ourselves to e.g.~integers.

Like with the first occurrence, this is trivially easy for sequential algorithms.
However the na\"ive constant-time algorithm, which checks for each index if there is no other index with a bigger entry, takes quadratic work.
Valiant improves this with an algorithm that takes almost-linear work.

\begin{lemma}\label{lem:epsSearch}\cite[Corollary 2]{valiantParallelismComparisonProblems1975}
	For any $\epsilon > 0$, there is a parallel constant-time algorithm that, given an array $X$ of size $n$, finds the maximum on a "common CRCW PRAM" with $\mathcal{O}(n^{1+\epsilon})$ processors.
\end{lemma}
\begin{proof}
	The algorithm runs over $\frac{1}{\epsilon}$ iterations.
	It uses a series of increasingly smaller arrays $X_i$ with $|X_i| = n^{1 - i\epsilon}$.
	The first iteration $i = 1$ begins with $X_0 = X$.
	In iteration $i$, the algorithm splits $X_{i-1}$ into segments of length $\lceil n^{\epsilon} \rceil$.
	Then, for each segment, it finds the maximum na\"ively.
	It collapses each segment into its maximum, thus creating the array $X_i$.
	
	It is easy to see that every $X_i$ contains the maximum of $X$.
	Since $|X_{\frac{1}{\epsilon}}| = 1$ the algorithm is left with only that maximum after $\frac{1}{\epsilon}$ iterations.
	Each iteration has at most $\frac{n}{n^{\epsilon}} = n^{1-\epsilon}$ segments and the na\"ive search for each segment runs on $\mathcal{O}(n^{2\epsilon})$ processors.
	Therefore the entire algorithm runs on $\mathcal{O}(n^{1+\epsilon})$ processors.
\end{proof}

\subparagraph{Prefix Sums}
The final static problem is the computation of prefix sums.
While a "common CRCW PRAM" can not calculate exact prefix sums in constant time with a polynomial number of processors (follows from \cite[Theorem 3.3]{furstParityCircuitsPolynomialtime1984} and \cite[Theorem 1]{stockmeyerSimulationParallelRandom1984}), it can calculate useful approximations.
We first define these approximations.

A $\epsilon$-approximate sum is within one and $(1+\epsilon)$-times of the true sum.
This notion carries over to $\epsilon$-approximate prefix sums.
Approximate prefix sums are consistent if the difference of each step from one prefix sum to the next is at least as big as the element at that position.

\begin{definition}\label{def:approxPrefixSum}
	Let $X$ be an array of $n$ integers and $\epsilon > 0$.
	An array $Y$ of $n$ integers contains consistent $\epsilon$-approximate prefix sums of $X$ if, for every $1 \leq i \leq n$, the following hold:
	\begin{enumerate}
		\item \textbf{Approximation:} Let $z_i = \sum_{j = 1}^{i} X[i]$. Then $z_i \leq Y[i] \leq (1+\epsilon)z_i$
		\item \textbf{Consistency:} $Y[i] - Y[i-1] \geq X[i]$.
	\end{enumerate}
\end{definition}

Goldberg and Zwick show that a "common CRCW PRAM" can calculate consistent approximate prefix sums efficiently in parallel constant time \cite[Theorem 4.2]{goldbergOptimalDeterministicApproximate1995}.

\begin{lemma}\label{lem:approxPrefixSum}
	\cite[Theorem 4.2]{goldbergOptimalDeterministicApproximate1995}
	For any $\epsilon > 0$ and $a  > 0$, there is a parallel constant-time algorithm that calculates consistent $(\log n)^{-a}$-approximate prefix-sums on a "common CRCW PRAM" with $\mathcal{O}(n^{1+\epsilon})$ processors.
\end{lemma}

We do not repeat the proof of \Cref{lem:approxPrefixSum} by Goldberg and Zwick \cite[Theorem 4.2]{goldbergOptimalDeterministicApproximate1995} here because it is quite complicated.

We continue with dynamic problems, starting with some further auxiliary algorithms before we tackle the problems described in \Cref{subsec:auxiliary}.

\subsection{String}

The first problem we look at is \dynamicString, which maintains a string with character insertion and deletion.
It forms one pillar of our solution for \dynamicMarkedString.

\begin{itemize}
	\item \AP \emph{Problem:} \intro *\dynamicString
	\item \emph{Initialization:} String $S = \epsilon$, alphabet $\Sigma = [1,n]$, maximum size $n$
	\item \emph{Updates:}
	\begin{itemize}
		\item $\textsc{Insert}(i,\sigma)$: Set $S$ to $S[1,i-1] \circ \sigma \circ S[i,|S|]$.
		\item $\textsc{Delete}(i)$: Set $S$ to $S[1,i-1] \circ S[i+1,|S|]$.
	\end{itemize}
	\item \emph{Queries:}
	\begin{itemize}
		\item $\textsc{Char}(i)$: Return $S[i]$.
		\item $\textsc{SubStr}(i,m)$: Return $S[i,i+m)$.
	\end{itemize}
\end{itemize}

A dynamic parallel constant-time algorithm can maintain \dynamicString~efficiently.

\begin{lemma}\label{lem:dynString}
	For any $\epsilon > 0$, there is a dynamic algorithm for \dynamicString~on a "common CRCW PRAM", where
	\begin{itemize}
		\item the size of the data structure is in $\mathcal{O}(n)$,
		\item the initialization runs in parallel constant time on $\mathcal{O}(n^{\epsilon})$ processors,
		\item updates run in parallel constant time on $\mathcal{O}(n^{\epsilon})$ processors, and
		\item queries run in parallel constant time on $\mathcal{O}(mn^{\epsilon})$ processors.
	\end{itemize}
\end{lemma}

\subparagraph{Data Structure}

The algorithm uses a tree where each inner node has up to $2n^{\epsilon}$ children, stored in a compact array.
Each node represents some substring.
Inner nodes represent the concatenation of the substrings of their children.
Leaves store single characters.

The distance between the root and a leaf is always $2\lceil \frac{1}{\epsilon} \rceil$.
The tree for the empty string is a path of length $2\lceil \frac{1}{\epsilon} \rceil-1$ from the root to an `inner' node with no leaves.

The number of children of inner nodes is adaptive.
This allows the algorithm to insert or delete characters anywhere without having to shift more than $\mathcal{O}(n^{\epsilon})$ children.

Inner nodes maintain the invariant that two adjacent siblings have in total at least $n^{\epsilon}$ children.
Every inner node maintains prefix sums over its children, counting the number of leaves in their subtrees.
The algorithm uses these prefix sums to resolve the index of each leaf.
The index of the character in a leaf is equal to the sum over all prefix sums along the path from that leaf to the root plus one.

See \Cref{fig:dyn_string} for an example.

\begin{figure}[ht]
	\centering \includegraphics{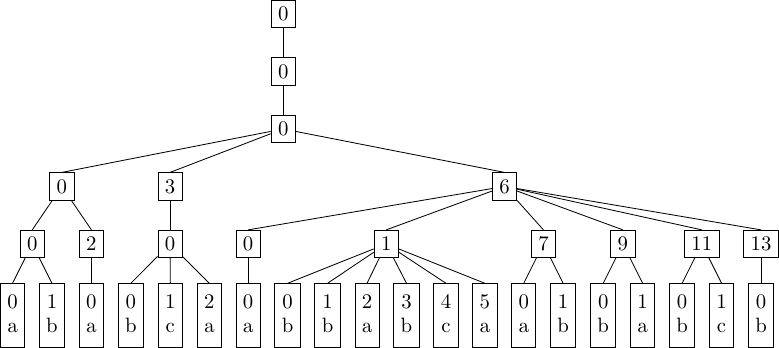}
	\caption{Possible tree for string $S = \text{ababcaabbabcaabbabcb}$ with $\epsilon = \frac{1}{3}, n = 27, n^{\epsilon} = 3$. (Note that $n$ is the \emph{maximum} length of $S$, not necessarily the \emph{actual} length.) Each node is labeled with its prefix sum, i.e.\ the total number of leaves in all its left siblings. Each leaf is also labeled with its character.}
	\label{fig:dyn_string}
\end{figure}

\subparagraph{Queries}
A query $\textsc{Char}(i)$ starts at the root.
It uses the prefix sums stored at the root to determine the child whose subtree contains the $i$-th character, then continues at that child until it reaches a leaf.

A query $\textsc{SubStr}(i,m)$ calls $\textsc{Char}$ a total of $m$ times for every index in $[i,i+m)$ in parallel and collects all results in an array.

\subparagraph{Updates}
To insert a character, the algorithm uses the procedure from \textsc{Char} to find the right spot and inserts a new leaf for this character.
It updates the prefix sums at every node along the path from the new leaf to the root, adding $1$ to every prefix sum that belongs to a node on the path or a right sibling of a node on the path.
If the parent of the new leaf exceeds its size constraint of at most $2n^{\epsilon}$ children, the algorithm splits it in half.
This might in turn cause its parent to grow to big, in which case this procedure travels upwards, splitting nodes until it reaches a node where the size constraint is met.
The root always meets this size constraint so long as the length of the string does not exceed $n$.
See \Cref{fig:dyn_string_insert} for an example.

\begin{figure}[ht]
	\centering \includegraphics{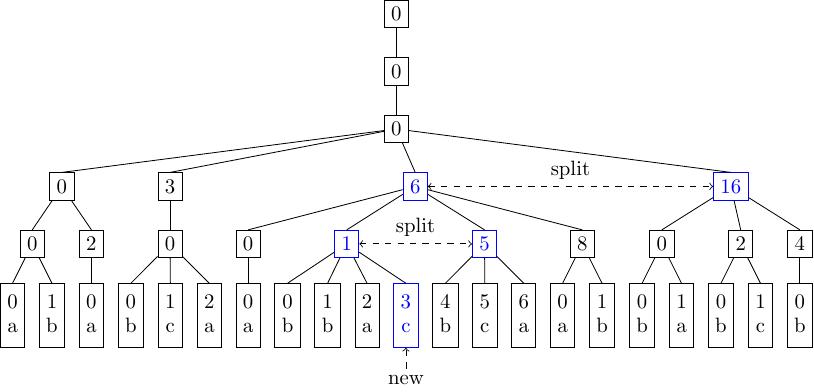}
	\caption{Tree from \Cref{fig:dyn_string} after update $\textsc{Insert}(11,\text{c})$. The new leaf and split inner nodes are marked.}
	\label{fig:dyn_string_insert}
\end{figure}

Deleting characters works similarly.
The only difference is that now the algorithm has to check the other size constraint, which demands that neighboring siblings have at least $n^{\epsilon}$ children in total.
If it encounters a pair of children that does not fit this size constraint, it merges these children.
See \Cref{fig:dyn_string_delete} for an example.

\begin{figure}[ht]
	\centering \includegraphics{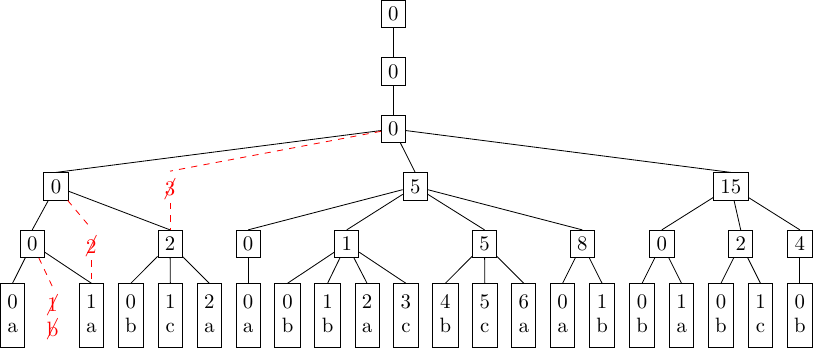}
	\caption{Tree from \Cref{fig:dyn_string_insert} after update $\textsc{Delete}(2)$. Deleted nodes are crossed out.}
	\label{fig:dyn_string_delete}
\end{figure}

\subsection{Nearest Neighbor}

The next algorithm solves the problem \dynamicNeighbor, which maintains successors and predecessors over marked elements of some fixed domain.
This algorithm was taken from \cite[Lemma 4.1]{schmidtWorksensitiveDynamicComplexity2021} and then extended to support batch updates.
The algorithm for \dynamicNeighbor~serves as the other pillar of our algorithm for \dynamicMarkedString.

\begin{itemize}
	\item \AP \emph{Problem:} \intro *\dynamicNeighbor
	\item \emph{Initialization:} Domain $\Delta = [1,n]$, subset $Z \subseteq \Delta$ with $Z = \emptyset$ initially
	\item \emph{Updates:}
	\begin{itemize}
		\item $\textsc{Mark}(X)$: Set $Z$ to $Z \cup X$.
		\item $\textsc{Unmark}(X)$: Set $Z$ to $Z \setminus X$.
	\end{itemize}
	\item \emph{Queries:}
	\begin{itemize}
		\item $\textsc{Succ}(i)$: Return $\min\{ j \in Z \mid j > i \}$, or $\bot$ if none exists.
		\item $\textsc{Pred}(i)$: Return $\max\{ j \in Z \mid j < i \}$, or $\bot$ if none exists.
	\end{itemize}
\end{itemize}

The size of the batches in our use-cases is very small with $|X| = \mathcal{O}(\log^* n)$.

\begin{lemma}\label{lem:fixedSuccessor}\cite[Lemma 4.1]{schmidtWorksensitiveDynamicComplexity2021}
	For any $\epsilon > 0$, there is a dynamic algorithm for \dynamicNeighbor~on a "common CRCW PRAM", where
	\begin{itemize}
		\item the size of the data structure is in $\mathcal{O}(n)$,
		\item the initialization runs in parallel constant time on $\mathcal{O}(n)$ processors,
		\item updates run in parallel constant time on $\mathcal{O}(\log n |X|^{2+\epsilon})$ processors, and
		\item queries run in parallel constant time on $\mathcal{O}(\log n)$ processors.
	\end{itemize}
\end{lemma}

The following dynamic algorithm is taken from Schmidt et al.\ \cite[Lemma 4.1]{schmidtWorksensitiveDynamicComplexity2021}, with some modifications for batch updates.

\subparagraph{Data Structure}
The algorithm uses a binary tree with one leaf for every element in the domain $\Delta = [1,n]$.
Every inner node represents some subset $[\alpha,\beta]$ of $\Delta$.

The inner node for subset $[\alpha,\beta]$ stores  $\min(Z \cap [\alpha,\beta])$ and $\max(Z \cap [\alpha,\beta])$.
The algorithm uses these min- and max-values to derive the successor and predecessor of a given $i$.

See \Cref{fig:dyn_neighbor} for an example.

\begin{figure}[ht]
	\centering \includegraphics[scale=0.95]{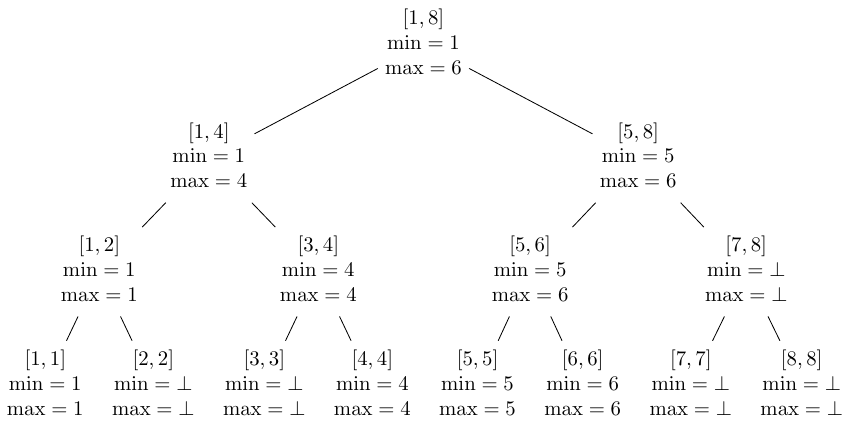}
	\caption{Tree for domain $\Delta = [1,8]$ and subset $Z = \{1,4,5,6\}$}
	\label{fig:dyn_neighbor}
\end{figure}

\subparagraph{Query}
The successor of an index $i$ is stored in a leaf somewhere to the right of the $i$-th leaf.
The algorithm can cover this entire search range by concatenating the intervals of just logarithmically many nodes.
The successor is the min-value of this concatenated interval.

To answer $\textsc{Succ}(i)$, the algorithm looks in parallel at every right sibling of a node along the path from the $i$-th leaf to the root.
The intervals of the right siblings form the concatenation mentioned above.
It uses \Cref{lem:linSearch} to identify the lowest sibling where the min-value is defined.
The lowest min-value is the min-value of the concatenation.
The successor is the min-value of this sibling.

Answering $\textsc{Prec}(i)$ works similarly.
The only difference is that the algorithm looks at the max-values of left siblings.

\subparagraph{Updates}
Marking or unmarking some $i$ affects all nodes on the path from the $i$-th leaf to the root.
For marking, the algorithm might have to replace the min-values or max-values stored at a node with a new, better value from $X$.
This new value is always the smallest/largest value in $X$ that belongs to this node.
For unmarking, it has to replace all min-values and max-values of nodes that store a value from $X$ with the next-best value from $Z \setminus X$.
The next-best value is always a successor/predecessor of a removed value.

To process $\textsc{Mark}(X)$, the algorithm looks in parallel at every node that lies on a path from a leaf for a value in $X$ to the root.
For every such node, it uses \Cref{lem:epsSearch} to determine the minimum value in $X$ that belongs to this node.
It overwrites the min-value stored at the node with the this new minimum if the new minimum is better.
It updates the max-value similarly.

For updates $\textsc{Unmark}(X)$, the algorithm first queries the successor and predecessor of every value in $X$.
For every value in $x \in X$, it looks at every node that lies on a path from the $x$-th leaf to the root, all in parallel.
At every node it visits, it checks if the min-value is in $X$.
If it is, it uses \Cref{lem:epsSearch} to find the smallest value in $X$ whose successor belongs to this node and is not in $X$.
This successor becomes the new min-value.
It updates max-values similarly.

\subsection{Marked String}

\Cref{lem:shiftingSuccessor} states that there is a constant-time algorithm for \dynamicMarkedString~on a "common CRCW PRAM" where updates run on $\mathcal{O}(n^{\epsilon} |X|^{2+\epsilon})$ processors and queries run on $\mathcal{O}(mn^{\epsilon})$ processors.
The problem \dynamicMarkedString~is defined as follows.

\begin{itemize}
	\item \emph{Problem:} \dynamicMarkedString
	\item \emph{Initialization:} Alphabet $\Sigma = [1,n]$, string $S = \epsilon$, marked indices $Z = \emptyset$, maximum size $n$
	\item \emph{Updates:}
	\begin{itemize}
		\item $\textsc{Insert}(i,\sigma): \begin{cases}
			\text{Set $S$ to $S[1,i-1] \circ \sigma \circ S[i,|S|]$.}\\
			\text{Set $Z$ to $\{ j \mid j \in Z, j < i \} \cup \{ j+1 \mid j \in Z, j \geq i \}$.}
		\end{cases}$
		\item $\textsc{Delete}(i): \begin{cases}
			\text{Set $S$ to $S[1,i-1] \circ S[i+1,|S|]$.}\\
			\text{Set $Z$ to $\{ j \mid j \in Z, j < i \} \cup \{ j-1 \mid j \in Z, j > i \}$.}
		\end{cases}$
		\item $\textsc{Mark}(X)$: Set $Z$ to $Z \cup X$.
		\item $\textsc{Unmark}(X)$: Set $Z$ to $Z \setminus X$.
	\end{itemize}
	\item \emph{Queries:}
	\begin{itemize}
		\item $\textsc{SubStr}(i,m)$: Return $S[i,i+m)$.
		\item $\textsc{Succ}(i)$: Return $\min\{ j \in Z \mid j > i \}$, or $\bot$ if none exists.
		\item $\textsc{Pred}(i)$: Return $\max\{ j \in Z \mid j < i \}$, or $\bot$ if none exists.
	\end{itemize}
\end{itemize}

Generally speaking, our algorithm for \dynamicMarkedString~combines \dynamicString~with \dynamicNeighbor.
It uses the tree from \dynamicMarkedString~and adds to each node the min- and max-values from \dynamicNeighbor.

\subparagraph{Data Structure}
The algorithm uses the tree from \dynamicString.
Inner nodes represent some substring of $S$ and have up to $2n^{\epsilon}$ children stored in a compact array.
There is a leaf for every character in $S$.

The distance between the root and a leaf is always $2\lceil \frac{1}{\epsilon} \rceil$ and the tree for $S = \epsilon$ is a length $2\lceil \frac{1}{\epsilon} \rceil - 1$ path from the root to an `inner' node with no leaves.

The total number of children of two neighboring inner nodes is always at least $n^{\epsilon}$.
Inner nodes maintain prefix sums over the sizes of the sub-intervals of their children.

A leaf for some character $S[i]$ stores if $i \in Z$.
An inner node for the substring $S[\alpha,\beta] \subseteq \Delta$ stores $\min(Z \cap [\alpha,\beta])$ and $\max(Z \cap [\alpha, \beta])$.

\subparagraph{Queries}
The query \textsc{SubStr} is identical to its counterpart in \dynamicString.

The queries \textsc{Succ} and \textsc{Pred} are very similar to \dynamicNeighbor.
The only difference is that, due to the adaptive nature of the tree, they travel through the levels in sequence via the same procedure as for \dynamicString.

To answer $\textsc{Succ}(i)$, the algorithm first travels to the leaf that stores $S[i]$.
It finds this leaf with the same procedure as was used for \dynamicString~to answer $\textsc{Char}(i)$.

It looks at all siblings to the right of this leaf.
If one of them stores a marked character, it return the index of the leftmost sibling with a marked character.

Otherwise, the algorithm travels bottom-up through the tree, stopping when it finds a level where the current node has a sibling to the right of it where the min-value is defined.
It returns the min-value of the leftmost sibling where it is defined.

It answers $\textsc{Pred}(i)$ similarly.

\subparagraph{Updates to the Domain}
The updates $\textsc{Insert}(i,\sigma)$ and $\textsc{Delete}(i)$ that modify the underlying string are almost identical to their counterparts in the algorithm for \dynamicString.
The only difference is that, when an update splits or merges inner nodes, it also has to update their min- and max-values.
This is trivial for merging nodes.
When it however splits one node into a left and right node, the algorithm has to calculate a new max-value for the left node and min-value for the right node.
The max-value of the left node is the max-value from the rightmost child of this node where the max-value is defined, and vice versa for the min-value of the right node.

\subparagraph{Updates to the Subset}
Updates $\textsc{Mark}(X)$ and $\textsc{Unmark}(X)$ that modify the subset $Z$ are very similar to their counterparts in the algorithm for \dynamicNeighbor.
The only difference is that, due to the adaptive nature of the tree, updates travel through the levels in sequence.

\subsection{Search Tree}

We describe a dynamic search tree which we will ultimately use in our algorithm for \dynamicNames.
The tree itself solves the problem \dynamicFunction.
The goal in \dynamicFunction~is to store a function that maps keys to values.
Each key is a length $m$ string over the alphabet $\Sigma = [1,n]$.
The value is from some domain $\Delta$ and we assume for simplicity's sake that each value fits in constant memory and can be copied in constant time.
Updates can add new keys, delete existing keys, and change the value of a key.
The value of a new key is undefined ($\bot$) initially.
\begin{itemize}
	\item \AP \emph{Problem:} \intro *\dynamicFunction
	\item \emph{Initialization:} Alphabet $\Sigma = [1,n]$, set $Z \subseteq \Sigma^m$ with $Z = \emptyset$ initially, function $f:Z \to \Delta$, string length $m$, maximum size $n$ of $Z$
	\item \emph{Updates:}
	\begin{itemize}
		\item $\textsc{SetValue}(x,y)$: If $x \in Z$, set $f(x)$ to $y$.
		\item $\textsc{Add}(X)$: Set $Z$ to $Z \cup X$. For each new $x \in X$, set $f(x) = \bot$.
		\item $\textsc{Remove}(X)$: Set $Z$ to $Z \setminus X$.
	\end{itemize}
	\item \emph{Queries:}
	\begin{itemize}
		\item $\textsc{GetValue}(x)$: Return $f(x)$, or $\bot$ if $x \notin Z$.
	\end{itemize}
\end{itemize}

Note that the domain $Z$ of function $f$ is not fixed.
It grows with $\textsc{Add}(X)$ and shrinks with $\textsc{Remove}(X)$.

In our use-cases, the batch size is $|X| = \mathcal{O}(\log^* n)$ and the string-length is $m = \mathcal{O}(n^{\epsilon})$.

\begin{lemma}\label{lem:searchTree}
	For any $\epsilon > 0$, there is a dynamic algorithm for \dynamicFunction~on a "common CRCW PRAM", where
	\begin{itemize}
		\item the size of the data structure is in $\mathcal{O}(n^{1+\epsilon})$,
		\item the initialization runs with constant work,
		\item updates run in parallel constant time on $\mathcal{O}(2^{|X|}mn^{\epsilon})$ processors,
		\item multiple updates $\bettersc{SetValue}(x,y)$ for different $x$ can run concurrently, and
		\item queries run in parallel constant time on $\mathcal{O}(mn^{\epsilon})$ processors.
	\end{itemize}
\end{lemma}

We use a similar tree to \Cref{lem:dynString}.
The main difference is that the children are sorted by their substrings.
This comes at the cost of no longer being able to use compact arrays.%
\footnote{Compact sorting is impossible in parallel constant time on a polynomial number of processors (follows from \cite[Theorem 3.3]{furstParityCircuitsPolynomialtime1984} and \cite[Theorem 1]{stockmeyerSimulationParallelRandom1984}).}
Instead children are stored in padded arrays.
A padded array may contain unused \textit{void} cells.
A padded array is sorted if removing every void cell produces a sorted (compact) array.

\subparagraph{Data Structure}
The algorithm uses a tree where inner nodes have up to $2n^{\epsilon}\log n$ children, stored in a padded array of size $2n^{\epsilon}\log n$.

The distance between the root and a leaf is always $2\lceil \frac{1}{\epsilon} \rceil$ and the tree for $Z = \emptyset$ is a length $2\lceil \frac{1}{\epsilon} \rceil - 1$ path from the root to an `inner' node with no leaves.

Similar to \dynamicString, the number of children is adaptive to limit the effect of insertions and deletions.
Unlike \dynamicString, the children are not stored compactly.
This is the cost of maintaining the sorting in constant time.

The total number of children of two neighboring inner nodes is always at least $n^{\epsilon}$.
Inner nodes keep track of the smallest and largest key in their subtree.
The children of each inner node are sorted so that the largest key of every node is smaller than the smallest key of that node's right sibling.

Each leaf stores a key $z \in Z \subseteq \Sigma^m$ and its value.

\subparagraph{Queries}
A query $\textsc{GetValue}(x)$ traverses the tree top-down.
Starting at the root, it looks for the child where the smallest key is at least $x$ and the largest key is at most $x$.
If such a child exists, the search continues at this child.
This repeats until it reaches a leaf, where $f(x)$ is stored, or can not find a child, in which case $x \notin Z$ and it returns $\bot$.

\subparagraph{Updates}
An update $\textsc{SetValue}(x,y)$ uses the same procedure as a query to find the leaf for key $x$, then exchanges the value stored there with $y$.

Updates $\textsc{Add}$ and $\textsc{Remove}$ travel through the tree bottom-up.
Every update first sorts $X$ into a compact array.
Compact sorting in constant-time requires an exponential number of processors, meaning that this step runs on $\mathcal{O}(2^{|X|}m)$ processors. 
Recall that in our use-cases, $|X| = \mathcal{O}(\log^* n)$.

An update $\textsc{Add}(X)$ travels to every lowermost inner node where values from $X$ are to be inserted.
The first obstacle to overcome is inserting new leaves into the sorted array of children.
The algorithm uses the sorted list $X$ to determine, for each child, how many new leaves are to be inserted directly behind it.
It then computes approximate consistent prefix-sums over these values via \Cref{lem:approxPrefixSum} to determine how far to shift each old leaf to make space for the new leaves.
This might result in an array of size larger than $2n^{\epsilon}\log n$.
The algorithm deals with the size constraints later.

The next task is updating the smallest and largest keys stored at inner nodes.
The algorithm visits every inner node along a path from a newly inserted leaf to the root and updates the smallest and largest keys accordingly.

The final task enforces the size constraints.
If an inner nodes has an array larger than $2n^{\epsilon}\log n$, the algorithm first computes how often this array has to be split so that the length of each part is between $n^{\epsilon}\log n$ and $2n^{\epsilon}\log n$.
It splits this node into that many parts and computes their smallest and largest keys.
To fit these new children into the array at the parent, the algorithm again uses approximate consistent prefix sums over its children to compute how far to shift each child.
This continues upwards until either all size-constraints are met or it reaches the root.

An update $\textsc{Remove}(X)$ travels to every leaf that contains a key in $X$ and deletes it.

It visits every node along a path from a deleted leaf to the root in bottom-up order and updates their smallest and largest keys by taking the minimum and maximum of the smallest and largest keys of their children.

Then, for every node where a child was deleted, it collects all children of children of this node into one array of size at most $4n^{2\epsilon}\log^2 n$, then runs an approximate consistent prefix-sum to redistribute them into children with approximately $2n^{\epsilon}\log n$ children each.%
\footnote{
	We do not use the same procedure as for \dynamicString~here because of the batch-updates.
	If we attempt to simply merge children with their siblings if they get to small, we run into issues when encountering many neighboring small children.
}

This procedure continues upwards until either all size-constraints are met or it reaches the root.
It checks both the upper and lower size-constraints because the redistributing can break both of them.

\subsection{Names}

\Cref{lem:names} states that there is a constant-time algorithm for \dynamicNames~on a "common CRCW PRAM" where queries run on $\mathcal{O}(mn^{\epsilon})$ and updates on $\mathcal{O}(2^kmn^{\epsilon})$ processors with $|X| \leq k$ and $\Delta = [1,nk]$.
The problem \dynamicNames~is defined as follows.

\begin{itemize}
	\item \emph{Problem:} \dynamicNames
	\item \emph{Initialization:} Alphabet $\Sigma$, domain $\Delta$, set $Z = \emptyset$, function $f: Z \to [1,n]$, injective function $g: Z \to \Delta$, maximum size $n$ of $Z$, substring size $m$.
	\item \emph{Updates:}
	\begin{itemize}
		\item $\textsc{Add}(X):\begin{cases}
			\text{For each unique $x \in X$, if $x \in Z$, set $f(x)$ to $f(x)+\setcount{X}{x}$.}\\
			\text{If $x \notin Z$ add $x$ to $Z$ with $f(x) = \setcount{X}{x}, g(x) \in \Delta$ so that $g$ is injective.}
		\end{cases}$
		\item $\textsc{Sub(X)}:\begin{cases}
			\text{For each unique $x \in X$, if $x \in Z$, set $f(x)$ to $f(x) - \setcount{X}{x}$.}\\
			\text{Then, for each $x \in X$, if $f(x) \leq 0$, remove $x$ from $Z$.}
		\end{cases}$
	\end{itemize}
	\item \emph{Queries:}
	\begin{itemize}
		\item $\textsc{Name}(x)$: Return $g(x)$, or $\bot$ if $x \notin Z$.
	\end{itemize}
\end{itemize}

In a nutshell the algorithm uses \dynamicFunction~to store $Z$ with $f$ and $g$.
It maintains which names are currently in use with an instance \dynamicNeighbor.

\subparagraph{Data Structure}
The algorithm uses an instance $\mathcal{F}$ of \dynamicFunction~from \Cref{lem:searchTree} and an instance $\mathcal{N}$ of \dynamicNeighbor~over the domain $\Delta = [1,nk]$.
The value stored in $\mathcal{F}$ for some $x$ is the tuple $(f(x),g(x))$.
In $\mathcal{N}$, all free names $i \in \Delta = [1,nk]$ are marked.
All names are marked initially.

\subparagraph{Queries}
Queries call the query from \dynamicFunction~to retrieve $(f(x),g(x))$, then return $g(x)$, if that value exists.

\subparagraph{Updates}
The batch $X$ arrives in the form of an array which might contain duplicate entries.
An update first determines all unique keys $x \in X$ and their number of occurrences $\setcount{X}{x}$.
This runs in parallel constant time on $\mathcal{O}(2^{|X|})$ processors.

An update $\textsc{Add}(X)$ first computes the new keys in $X$ which do not occur in $\mathcal{F}$ yet.
It assigns numbers $1,2,\ldots$ to each new key.
It uses $\mathcal{N}$ and these numbers to assign names to new keys.
For each new unique key $x$ with number $i$, it assigns the successor $\textsc{Succ}((i-1)n)$ in $\mathcal{N}$ as the name, then unmarks that name in $\mathcal{N}$.
It never assigns the same name to two different keys so long as there are never more than $n$ names in use at the same time.

Then it adds all new keys $x$ to $\mathcal{F}$, sets the name $g(x)$ to the name that was just assigned and initializes $f(x)$ with zero.

Finally, for every unique (new or old) key $x$, it increases the value of $f(x)$ stored in $\mathcal{F}$ by $\setcount{X}{x}$.

An update $\textsc{Sub}(X)$ begins by, for every unique key $x \in X$, reducing the value $f(x)$ stored in $\mathcal{F}$ by $\setcount{X}{x}$.
Then, for every $x \in X$ with $f(x) \leq 0$, it removes $x$ from $\mathcal{F}$ and marks  $g(x)$ in $\mathcal{N}$.

\section{Proof of \texorpdfstring{\Cref{lem:queryNoMuddling}}{Lemma 6}}\label{sec:queryNoMuddling}
\Cref{lem:queryNoMuddling} states that a parallel constant-time algorithm can answer substring equality on $\mathcal{O}(\log n)$ processors if, for each $\tau = 2^z, z \in \mathbb{N}, \tau \in [1,n]$, it has access to "string synchronizing sets" $(B^{\tau},f^{\tau})$ and successor- and predecessor-links that point to the closest element in $B^{\tau}$.

\begin{proof}
	The general idea is to cover $S[i,i+m]$ with few occurrences from the "string synchronizing sets".
	Then $S[i,i+m] = S[j,j+m]$ holds if and only if it is possible to cover $S[j,j+m]$ the same way.
	
	The algorithm runs in parallel over the levels $\tau$ for each $\tau = 2^z, z \in \mathbb{N}, \tau \in [1,n]$.
	The objective of level $\tau$ is to verify if $S[i+\frac{\tau}{2},i+\tau) = S[j+\frac{\tau}{2},j+\tau)$ and $S[i+m-\tau,i+m-\frac{\tau}{2}) = S[j+m-\tau,j+m-\frac{\tau}{2})$ hold.
	The check might also include up to $\frac{\tau}{2}$ characters on either side of these intervals.
	
	If $S[i+\frac{\tau}{2},i+\tau) \neq S[j+\frac{\tau}{2},j+\tau)$ or $S[i+m-\tau,i+m-\frac{\tau}{2}) \neq S[j+m-\tau,j+m-\frac{\tau}{2})$, then this level has to reject.
	If $S[i,i+\frac{3}{2}\tau) = S[j,j+\frac{3}{2}\tau)$ and $S[i+m-\frac{3}{2}\tau, i+m) = S[j+m-\frac{3}{2}\tau, j+m)$, then this level has to accept.
	Otherwise, this level can do either.
	
	There is some special handling for the first and last level.
	The first level $\tau = 1$ verifies $S[i] = S[j]$ and $S[i+m-1] = S[j,j+m-1]$ na\"ively and does nothing else.
	The last level is the maximum $\tau$ such that $\frac{3}{2}\tau \leq m$ holds.
	If $2\tau < m$ holds for this $\tau$, then there is a gap between $S[i+\frac{\tau}{2},i+\tau)$ and $S[i+m-\tau,i+m-\frac{\tau}{2})$.
	The last level closes this gap with two additional checks that are shifted by $\frac{\tau}{2}$ indices.
	It therefore also verifies that $S[i+\tau,i+\frac{3}{2}\tau) = S[j+\tau,j+\frac{3}{2}\tau)$ and $S[i+m-\frac{3}{2}\tau,i+m-\tau) = S[j+m-\frac{3}{2}\tau,j+m-\tau)$ hold.
	
	\begin{figure}[ht]
		\centering\includegraphics[scale=0.95]{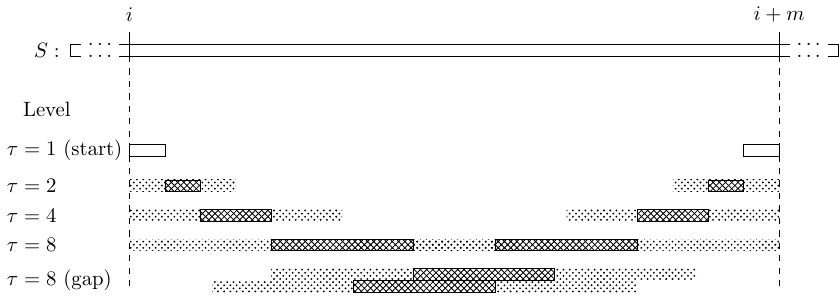}
		\caption{Covering the substring $S[i,i+m]$ with occurrences from the "synchronizing sets"}
		\label{fig:query_all_levels}
	\end{figure}
	See \Cref{fig:query_all_levels} for how the levels taken together cover the entire substring.
	The figure shows both the characters that are guaranteed to be part of the check (the level has to reject if they don't match) as well as characters that might or might not be included (the level has to accept if they match).
	
	We only describe the check for $S[i+\frac{\tau}{2},i+\tau) = S[j+\frac{\tau}{2},j+\tau)$.
	The other side is analogous.
	
	Level $\tau$ uses the successor-links to find the first occurrence from $B^{\tau}$ that begins at or after index $i$.
	If there are no periodic segments, then the first occurrence begins at some index $\in [i,i+\frac{\tau}{2})$ due to the "density"-property of $B^{\tau}$.
	\begin{figure}[ht]
		\centering\includegraphics[scale=0.95]{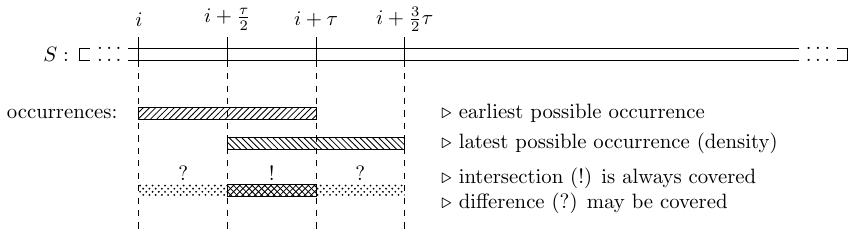}
		\caption{Indices covered by the first occurrence in $S[i,i+m]$ in the non-periodic case}
		\label{fig:query_single_level}
	\end{figure}
	See \Cref{fig:query_single_level} for the indices covered by the first occurrence.
	
	We distinguish three cases.
	
	In the first case, if the first occurrence contains $S[i+\frac{\tau}{2},i+\tau)$ and lies entirely in $S[i,i+m)$, then it looks for a corresponding occurrence with the same name for $S[j+\frac{\tau}{2},j+\tau)$.
	The check is successful if that corresponding occurrence exists.
	
	In the second case, if the candidate contains $S[i+\frac{\tau}{2},i+\tau)$ but does not lie entirely in $S[i,i+m)$, then $S[i,i+m)$ is too short for this particular level.
	This level does nothing and leaves the checks to the levels below it.
	
	In the third case, if the candidate does not contain $S[i+\frac{\tau}{2},i+\tau)$, then $S[i,i+m)$ starts with a periodic segment.
	The algorithm verifies that $S[j,j+m)$ has a matching periodic prefix by verifying that there are no occurrences there either.
	Level $\tau$ does not have to verify the period itself because there is a lower level $\tau' < \tau$ such that the period is short enough to fit into the substrings at that level, but long enough not to trigger the exception to the "density"-property.
	
	\begin{figure}[ht]
		\centering\includegraphics{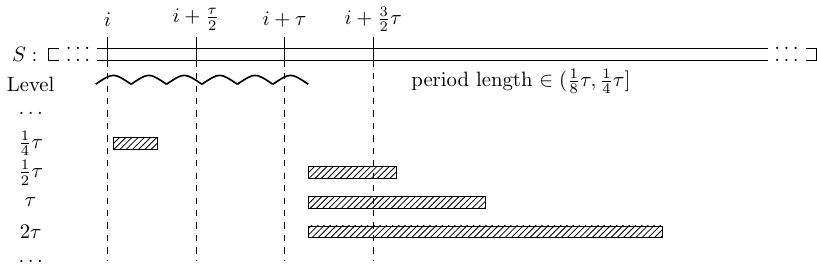}
		\caption{Example for a substring $S[i,i+m]$ with periodic prefix}
		\label{fig:query_periodic}
	\end{figure}
	
	See \Cref{fig:query_periodic} for an example.
	The substring $S[i,i+m)$ starts with a periodic prefix.
	The figure shows the candidate of four levels.
	Level $\tau' = \frac{1}{4}\tau$, which falls into the first case, guards the period.
	Note that its candidate lies entirely within the periodic prefix and is not shorter than the period.
	There is always a level for which this is the case.
	Level $2\tau$ guards the end and therefore the length of the periodic prefix and also falls into the first case.
	Levels $\tau$ and $\frac{1}{2}\tau$, which fall into the third case, only guard that there is a periodic prefix.
\end{proof}

\section{Proof of \texorpdfstring{\Cref{lem:prep}}{Lemma 10}}\label{sec:preparation}

\Cref{lem:prep} establishes the properties of the string of factor names produced by the preparation.
We first describe the preparation in more detail and then show how \Cref{lem:prep} follows from it.

\subsection{Description of the Preparation}

We give a more detailed description of how the algorithm prepares the processing of level $\ell$.
The preparation of the factors in $D^{\ell-1}$ consists of four steps which are processed sequentially in that order.

\subparagraph{Merging Short Factors}
The first step looks for factors shorter than $2^{\ell-1}$.
Note that, due to \Cref{lem:factorSize}, these short factors can not neighbor each other.
It first tries to merge every short factor with its left neighbor if their total length is at most $2^{\ell+1}$.
After this, it similarly tries to merge every remaining short factor with its right neighbor.

\subparagraph{Merging Runs}
The second step uses the successor- and predecessor-links in $\mathcal{D}^{\ell-1}$ to find runs of identical factors and merge each run into one factor.

\subparagraph{"Deactivating" Long Factors}
The third step "deactivates" all factors longer than $2^{\ell}$.
It "deactivates" a factor by replacing the name with $-1$ for the next actions.
"Deactivated factors" will not be merged further.

\subparagraph{Local Copies}
The final step creates local copies of the "relevant segments" of $\mathcal{D}^{\ell-1}$ and $\mathcal{S}^{\ell-1}$ to avoid conflicts with later actions.

\subsection{Proof}

\Cref{lem:prep} states that, if $d_1d_2\ldots$ is the string of factor names produced by the preparation of level $\ell$ and $v_1,v_2,\ldots$ are the corresponding factors, the following hold for every factor name in $d_1d_2\ldots$:

\begin{enumerate}
	\item A factor name $d_i$ is $-1$ if and only if the corresponding factor $v_i$ is longer than $2^\ell$.
	\item Two factor names $d_i$, $d_j$ that are not $-1$ are equal if and only if the corresponding factors $v_i, v_j$ are equal.
	\item A factor $v_i$ is only shorter than $2^{\ell-1}$ if the factor name $d_{i-1}$ before and the name $d_{i+1}$ after it are both $-1$.
	\item Neighboring factor names $d_i,d_{i+1}$ are distinct or both $-1$.
\end{enumerate}

\begin{proof}
We prove the four statements in \Cref{lem:prep} in the order they appear.

According to the final step, the names of all factors longer than $2^{\ell}$ are replaced with $-1$, proving the first statement of \Cref{lem:prep}.

The second statement follows from the fact that the algorithm uses $\mathcal{N}_D^{\ell-1}$ to assign consistent names to the factors in $\mathcal{D}^{\ell-1}$.

To prove the third statement, first assume that there is a factor $v_i$ in $D^{\ell-1}$ that is shorter than $2^{\ell}$.
Due to \Cref{lem:factorSize}, the combined size of neighboring factors in $D^{\ell-1}$ is at least $2^{\ell+1}$.
The neighbors of $d_i$ therefore have length greater than $2^{\ell}$.
The first step selects each factor shorter than $2^{\ell}$ and tries to merge them into their left or right neighbor.
Because short factors can not neighbor each other and because left- and right-merges are done sequentially, there are no conflicts during this step.
A short factor is not merged with a neighbor only if their combined length exceeds $2^{\ell+1}$.
This can only occur if both neighbors have length greater than $2^{\ell}$, meaning that their names will be replaced by $-1$ in the final step.

The fourth statement follows from the merging of runs.
There can not be any sequence $d_i\ldots d_i$ of equal names after this step because this step merges runs of equal factors.
There are also no later steps that could introduce new factor names into the sequence, except for $-1$ which is excluded in this statement.
\end{proof}

\section{Proof of \texorpdfstring{\Cref{lem:coinFlipping}}{Lemma 11}}\label{sec:coinFlipping}

\Cref{lem:coinFlipping} states that there is a function $f$ that, given a string $S$ over the alphabet $[-1,N-1]$, computes a bit-string of the same length, such that
\begin{enumerate}
	\item the bit $f(S)[i]$ depends only on $S[i-4,i+\log^* N + 5]$,
	\item the character $S[j]$ does not affect bit $f(S)[i]$ if there is a character $-1$ in between,
	\item if $S[i]$ is $-1$, then $f(S)[i,i+1] = 11$,
	\item the bit-string does not contain $11$, unless the previous property applies,
	\item the bit-string does not contain $0000$, and
	\item this function can be computed in $\mathcal{O}(\log^* n)$ time on $n$ processors.
\end{enumerate}

\begin{proof}
	The algorithm treats $S$ as the $N$-coloring of a path.
	In a nutshell, it reduces the number of colors first to $6$ over $\log^*N$ rounds, then to $3$, and finally sets all indices to $1$ where the $3$-coloring has a local minimum.
	
	The algorithm also has some special handling for indices $i$ with $S[i] = -1$.
	The intention is that $S[i] = -1$ acts as a barrier.
	The coloring to the left of $S[i] = -1$ is independent from the characters to the right of it and vice versa.
	The coloring of $S[i] = -1$ itself is irrelevant because the final result $f(S)[i]$ is fixed according to \cref{lem:coinFlipping}, Property \ref{enum:coinFlip:lowerforce}.
	
	To get from an $N$-coloring to a $6$-coloring, the algorithm repeatedly compares, for each index, the binary representation of its color with the binary representation of the next color.
	It determines the first index where the binary representations differ and sets the new color to the bit at that index followed by the index itself.
	This reduces the size of the binary representation from $b$ bits to $1 + \lceil\log b\rceil$ bits.
	Observe that this results in a valid coloring.
	
	For indices $S[i] = -1$, it always keeps the value $-1$.
	When comparing binary representations, it treats the binary representation of $-1$ as $2^b$ (i.e. every bit is different compared to a regular binary representation).
	This special handling of the color $-1$ ensures Property \ref{enum:coinFlip:localforce}.
	
	After $\log^* N$ iterations every value is either in $[0,5]$, which forms a $6$-coloring over these indices, or it is $-1$.
	The algorithm then looks at each index with color $3$ and assigns for each index the smallest available color from $\{0,1,2\}$ that does not occur at a neighbor.
	It repeats this procedure for $4$ and finally for $5$.
	
	The algorithm sets $f(S)[i] = 1$ if and only if at least one of the following holds:
	\begin{enumerate}
		\item $S[i] = -1$,
		\item $S[i-1] = -1$
		\item the final color at $i$ is $0$, or
		\item the final color at $i$ is $1$ and no neighbor has color $0$.
	\end{enumerate}
	
	The first two conditions ensure that Property \ref{enum:coinFlip:lowerforce} holds.
	The last two conditions simply check for a local minimum, ignoring the color $-1$.
	Since the colors form a valid 3-coloring, local minima can not neighbor each other, so Property \ref{enum:coinFlip:lower} holds, if we assume there are is no $-1$.
	Since there are only three colors, at least every fourth color has to be a local minima (if the coloring goes $0,1,2,1,0$), so Property \ref{enum:coinFlip:upper} holds as well.
	
	The algorithm concludes with a clean-up step to address the problem that Property \ref{enum:coinFlip:lower} might otherwise fail next to indices with $-1$.
	If $S[i] = -1, S[i+1] \neq -1$ the algorithm might produce $f(S)[i+1,i+2] = 11$.
	To fix this, if $f(S)[i+2,i+4] = 100$, it sets $f(S)[i+2,i+4]$ to $010$.
	Otherwise, if $f(S)[i+2,i+4] = 101$, it sets $f(S)[i+2,i+4]$ to $001$.
	The algorithm does the same mirrored for $S[j] = -1, S[j-1] \neq -1$ with $f(S)[j-1,j] = 11$.
	
	The algorithm runs over $\log^* N$ rounds where, for each round, it only compares each color with its neighbors.
	Property \ref{enum:coinFlip:work} follows from this observation.
	
	Finally, adding up all indices that affect one bit $S[i]$ yields the bounds in Property \ref{enum:coinFlip:localdist}.
\end{proof}

\section{Proof of \texorpdfstring{\Cref{lem:segments}}{Lemma 9}}\label{sec:proof:lem:segments}

\Cref{lem:segments} states that the "affected segment" for each index consists of the factor that contains that index plus $\log^* n+9$ factors before and $5$ factors after it.
It also states that the factors required to compute the "affected segment" are the factor that contains the index plus $\log^* n+14$ factors before and after it.

\begin{proof}
	Let $\Dell $ be a given decomposition and let $D^{\ell+1}$ be the next decomposition in the decompositions hierarchy.
	We calculate which factors in $\Dell $ affect a given factor $u$ in $D^{\ell+1}$.
	
	For some factor $u$ in $D^{\ell+1}$, let $v_i$ be the factor in $\Dell $ that has the same starting index.
	Our goal is to calculate $\alpha,\beta$ so that $u$ is affected by $v_{i-\alpha},\ldots,v_{i+\beta}$.
	We go backwards through the algorithm and count the factors that affect $u$.
	
	The final step is the level finalization.
	It merges factors according to the bit-string produced by the previous step.
	A merged factor begins at every factor of $\Dell $ where the bit-string is $1$.
	Since $v_j$ is the first factor of $\Dell $ in $u$, we know that the bit-string at $v_i$ is $1$.
	How many factors from $\Dell $ make up $u$ depends on how many $0$'s follow in the bit-string.
	Due to \Cref{lem:coinFlipping}, Property \ref{enum:coinFlip:upper}, it follows that there are at most three $0$'s after that $1$.
	The number of factors therefore depends on the bits for $v_{i+1},v_{i+2},v_{i+3}$.
	
	Before the finalization comes choosing the merges via the deterministic coin flipping in \Cref{lem:coinFlipping}.
	The result of the deterministic coin flipping depends on the sequence of factor names where some factors are "deactivated".
	According to \Cref{lem:coinFlipping}, Property \ref{enum:coinFlip:localdist}, the bit for some factor $v_j$ depends at most on the factor names of $v_{i-4},\ldots,v_{i+\log^* N + 5}$.
	
	Adding this to our previous bounds yields that $u$ depends on the factor names for $v_{i-4},\ldots,v_{i+\log^* N + 8}$.
	
	The sequence of factor names that is fed into the deterministic coin flipping is the result of the level preparation.
	
	The last step of the level preparation "deactivates" long factors.
	This has indirect effect on our bounds.
	Combined with \Cref{lem:coinFlipping}, Property \ref{enum:coinFlip:lowerforce}, it follows that we only have to worry about at most one long factor before and one long factor after $v_j$.
	Everything that comes after the first long factor cannot affect $u$.
	
	The earlier preparation steps create long factors.
	First observe that, due to \Cref{lem:factorSize}, if a factor is merged during the preparation, the resulting factor is longer than $2^{\ell}$ and will be "deactivated".
	We therefore only have to worry about the first factors on either side that are merged due to the preparation.
	For both runs and short factors, the merging depends depends on the factor directly before and directly after a given factor.
	This extends our bounds to $v_{i-5},\ldots,v_{i+\log^* n + 9}$.
	
	We have reached the beginning and conclude that the factors that affect $u$ are the $5$ factors before $v_i$ and the $\log^* n +9$ factors after it.
	
	To derive the size of the "affected segments", i.e. how many factors are affected by changes to $v_i$, we mirror this and arrive at the bounds stated in \Cref{lem:segments}.
	
	For the number of factors required to compute the "affected segment", we take the union of all factors that affect a factor in the "relevant segment".
	The furthest affected factor to the left is $v_{i-(\log^* n + 9)}$, which is in turn affected by $5$ factors before it.
	Thus the "relevant segment" extends $\log^* n + 14$ factors to the left.
	By similar logic it extends $\log^* n + 14$ factor to the right.
\end{proof}

\section{Proof of \texorpdfstring{\Cref{the:dyck1}}{Theorem 14}}\label{sec:dyck1}
Schmidt et al.\ show that membership to the "Dyck language" $D_1$ can be maintained in parallel constant time on $\mathcal{O}((\log n)^3)$ processors \cite[Theorem 6.5]{schmidtWorksensitiveDynamicComplexity2021}.
\Cref{the:dyck1} states that we can improve this bound to $\mathcal{O}(\log n \log \log n)$.

The problem discussed in this section does not allow character insertion or deletion and instead uses a ""void character"" $\bot$, which marks positions that should be ignored.
To deal with the "void character" more easily, we define an ""erasing-function"" $\erase{a}(S) $ that deletes all occurrences of character $a$ from a string $S$, which we use with $a = \bot$ to erase "void characters".

\begin{definition}\label{def:erase_func}
	The function $\intro *\erase{a}: \Sigma^* \mapsto (\Sigma \setminus \{a\})^*$ erases all occurrences of a symbol $a$ from a string with
	\[ \erase{a}(S) = \begin{cases}
		\epsilon & \mid S \in \{\epsilon, a\}\\
		\sigma & \mid S \in \Sigma \setminus \{a\}\\
		\erase{a}(u) \circ \erase{a}(v) & \mid S = u \circ v,\ u,v \in \Sigma^+\\
	\end{cases} \]
\end{definition}

We first recall the problem \member{$D_1$}~as it is defined by Schmidt et al.\ \cite[Section 2]{schmidtWorksensitiveDynamicComplexity2021}.
Let $\langle, \rangle$ be the parentheses from $D_1$.
The problem \member{$D_1$}~also uses the "void character" $\bot$ which is ignored by $D_1$.

\begin{itemize}
	\item \AP \emph{Problem:} \intro *\member{$D_1$}
	\item \emph{Initialization:} String $S=\bot^n$
	\item \emph{Updates:}
	\begin{itemize}
		\item $\textsc{Set}(i,\sigma)$: If  $S[i] = \bot$, set $S[i]$ to $\sigma \in \{ \langle, \rangle \}$.
		\item $\textsc{Reset}(i)$: Set $S[i]$ to $\bot$.
	\end{itemize}
	\item \emph{Query:}
	\begin{itemize}
		\item $\textsc{Member}()$: Is $\erase{\bot}(S) \in D_1$?
	\end{itemize}
\end{itemize}

Note that, unlike most problems discussed before, \member{$D_1$}~does not allow inserting or deleting indices from the string and instead uses a "void character" $\bot$.

The algorithm is in large parts identical to the original algorithm described by Schmidt et al.\ \cite[Theorem 6.5]{schmidtWorksensitiveDynamicComplexity2021}.
The only difference lies in one step of the updates, which we will point out when we come to it.

\subparagraph{Data Structure}
\begin{figure}[ht]
	\centering \includegraphics{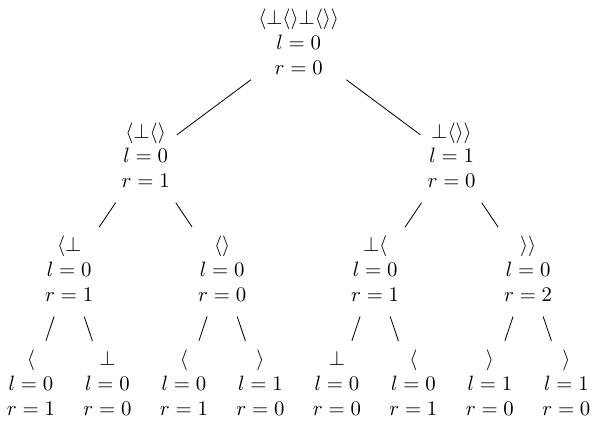}
	\caption{Data Structure for $S = \langle \bot \langle \rangle \bot \langle \rangle \rangle$}
	\label{fig:d1_structure}
\end{figure}

The algorithm uses a binary tree with one leaf for every character of the underlying string $S$.
Every inner node represents some substring $S[\alpha,\beta)$.

The inner node for substring $S[\alpha,\beta]$ stores the number of unmatched parentheses of this substring.
It stores the number of unmatched closing parentheses $l(S[\alpha, \beta))$ and the number of unmatched opening parentheses $r(S[\alpha,\beta])$.

See \Cref{fig:d1_structure} for an example.

\subparagraph{Updates}

An update $\textsc{Set}(i,\sigma)$ or $\textsc{Reset}(i)$ clearly only affects nodes along the path from the root to the $i$-th leaf, which we call the ""affected path"".

\begin{table}[ht]
	\centering\begin{tabular}{l|l|l}
		$\textsc{Set}(i,\langle)$&$i \in [\alpha,\gamma) $&$i \in [\gamma,\beta)$\\ \hline
		$r(S[\alpha,\gamma)) < l(S[\gamma,\beta))$ & decrease $l(S[\alpha,\beta))$ by $1$&same as for $[\gamma,\beta)$\\ \hline
		$r(S[\alpha,\gamma)) \geq l(S[\gamma,\beta))$&same as for $[\alpha,\gamma)$&increase $r(S[\alpha,\beta))$ by $1$
	\end{tabular}\\[1em]
	\centering\begin{tabular}{l|l|l}
		$\textsc{Set}(i,\rangle)$&$i \in [\alpha,\gamma) $&$i \in [\gamma,\beta)$\\ \hline
		$r(S[\alpha,\gamma)) \leq l(S[\gamma,\beta))$ & increase $l(S[\alpha,\beta))$ by $1$&same as for $[\gamma,\beta)$\\ \hline
		$r(S[\alpha,\gamma)) > l(S[\gamma,\beta))$&same as for $[\alpha,\gamma)$&decrease $r(S[\alpha,\beta))$ by $1$
	\end{tabular}\\[1em]
	\centering\begin{tabular}{l|l|l}
	$\textsc{Reset}(i),\ S[i] = \langle$&$i \in [\alpha,\gamma) $&$i \in [\gamma,\beta)$\\ \hline
	$r(S[\alpha,\gamma)) \leq l(S[\gamma,\beta))$ & increase $l(S[\alpha,\beta))$ by $1$&same as for $[\gamma,\beta)$\\ \hline
	$r(S[\alpha,\gamma)) > l(S[\gamma,\beta))$&same as for $[\alpha,\gamma)$&decrease $r(S[\alpha,\beta))$ by $1$
	\end{tabular}\\[1em]
	\begin{tabular}{l|l|l}
		$\textsc{Reset}(i),\ S[i] = \rangle$&$i \in [\alpha,\gamma) $&$i \in [\gamma,\beta)$\\ \hline
		$r(S[\alpha,\gamma)) < l(S[\gamma,\beta))$ & decrease $l(S[\alpha,\beta))$ by $1$&same as for $[\gamma,\beta)$\\ \hline
		$r(S[\alpha,\gamma)) \geq l(S[\gamma,\beta))$&same as for $[\alpha,\gamma)$&increase $r(S[\alpha,\beta))$ by $1$
	\end{tabular}
	\caption{Effect of updates on the node with substring $S[\alpha,\beta), i \in [\alpha,\beta)$ where the children have substrings $S[\alpha,\gamma)$ and $S[\gamma,\beta)$}
	\label{tab:d1}
\end{table}
\Cref{tab:d1} shows the effects of an update on a node on the "affected path".
A node is ""effect-preserving"" if the effect for this node is defined as being the same as the effect for its child on the "affected path".
Otherwise it is ""effect-inducing"".
Observe that the upper left and lower right corners of each table are "effect-inducing" and the others are "effect-preserving".

The algorithm first identifies all "effect-inducing" nodes and updates them.
This runs in constant time on one processors per node, or $\mathcal{O}(\log n)$ processors for all nodes on the "affected path".

For "effect-preserving" nodes, it suffices to find the closest "effect-inducing" descendant and copy the effect from it.
The difficulty therefore lies in finding the closest "effect-inducing" descendant.
Schmidt et al. use a na\"ive search for this, which requires $\mathcal{O}(\log^2 n)$ processors per node or $\mathcal{O}(\log^3 n)$ processors in total.
We can improve this to $\mathcal{O}(\log n \log \log n)$ processors.

Let $A$ be a bit-array of size $\log n$ where the $i$-th bit is $1$ if and only if the $i$-th node on the "affected path" is "effect-inducing", with the root as the first node.

The algorithm computes successor-links for $A$.
This procedure is a static version of \Cref{lem:fixedSuccessor}.
The algorithm builds a binary tree over $A$ with a leaf for every $i \in [1,\log n]$.
Every node represents some interval $[\alpha,\beta] \subseteq [1,\log n]$.
It computes $\min\{ i \in [\alpha, \beta] \mid A[i] = 1 \}$ for every node with interval $[\alpha,\beta]$ using \Cref{lem:linSearch}.
This runs on $\mathcal{O}(\log n \log \log n)$ processors because the sum total length of all intervals of all nodes at one level is $|A| = \log n$ and there are $\log |A| = \log \log n$ levels. 
Retrieving successor-links from this tree works the same as successor-queries in \Cref{lem:fixedSuccessor}.
Each query runs on $\mathcal{O}(\log |A|) = \mathcal{O}(\log \log n)$ processors, or $\mathcal{O}(\log n \log \log n)$ for the entire array $A$.

The closest descendant of the $i$-th node on the "affected path" is the successor of $i$ in $A$.
Given the successor-links that the algorithm just computed, the complexity of finding the closest "effect-inducing" descendant of each "effect-preserving" node thus drops to $\mathcal{O}(1)$.

The overall complexity is dominated by the calculation of the successor-links, which runs on $\mathcal{O}(\log n \log \log n)$ processors.

\subparagraph{Queries}
The string $S$ is in $D_1$ if and only if there are no unmatched parentheses, meaning that $l(S) = r(S) = 0$.
The values $l(S)$ and $r(S)$ are retrieved from the root.

\section{Proof of \texorpdfstring{\Cref{cor:dyck}}{Corollary 15}}\label{sec:proof:dyck}

We prove that a dynamic parallel constant-time algorithm can maintain membership in $D_k$ on a "common CRCW PRAM" with $\mathcal{O}(n^{\epsilon})$ processors.

\subsection{Maintaining "Dyck Languages"}

The "Dyck language" $D_k$ is the language of well-balanced parentheses expressions with $k$ types of parentheses $\langle_1,\rangle_1,\ldots,\langle_k,\rangle_k$.
We first define the dynamic problem of maintaining membership in $D_k$.
We use the "erasing-function" $\erase{\bot}:(\Sigma\cup\{\bot\})^* \to \Sigma^*$ (\Cref{def:erase_func}) to erase occurrences of the "void character" $\bot$ from a string.

\begin{itemize}
	\item \AP \emph{Problem:} \member{$D_k$}
	\item \emph{Initialization:} String $S=\bot^n$
	\item \emph{Updates:}
	\begin{itemize}
		\item $\textsc{Set}(i,\sigma)$: If  $S[i] = \bot$, set $S[i]$ to $\sigma \in \{ \langle_j, \rangle_j \mid j \in [1,k] \}$.
		\item $\textsc{Reset}(i)$: Set $S[i]$ to $\bot$.
	\end{itemize}
	\item \emph{Query:}
	\begin{itemize}
		\item $\textsc{Member}()$: Is $\erase{\bot}(S) \in D_k$?
	\end{itemize}
\end{itemize}

Schmidt et al.\ prove a link between \member{$D_k$}~and the problem \stringEquality~\cite[Lemma 6.9]{schmidtWorksensitiveDynamicComplexity2021}, which maintains if two dynamic strings are equal.

\begin{itemize}
	\item \AP \emph{Problem:} \intro *\stringEquality
	\item \emph{Initialization:} Strings $S_1=\bot^n$, $S_2 = \bot^n$
	\item \emph{Updates:}
	\begin{itemize}
		\item $\textsc{Set}_1(i,\sigma)$: If  $S_1[i] = \bot$, set $S_1[i]$ to $\sigma$.
		\item $\textsc{Reset}_1(i)$: Set $S_1[i]$ to $\bot$.
		\item $\textsc{Set}_2(i,\sigma)$: If  $S_2[i] = \bot$, set $S_2[i]$ to $\sigma$.
		\item $\textsc{Reset}_2(i)$: Set $S_2[i]$ to $\bot$.
	\end{itemize}
	\item \emph{Query:}
	\begin{itemize}
		\item $\textsc{Equal}()$: Is $\erase{\bot}(S_1) = \erase{\bot}(S_2)$? 
	\end{itemize}
\end{itemize}

Schmidt et al.\ prove that, given an algorithm for \stringEquality, it is possible to solve \member{$D_k$}.
This increases the number of processors by a logarithmic factor.

\begin{lemma}\label{lem:StringEqualityToDyck} \cite[Lemma 6.9]{schmidtWorksensitiveDynamicComplexity2021}
	If there is a dynamic parallel constant-time algorithm that maintains \stringEquality~on a "common CRCW PRAM" with $W(n)$ processors, then there is a dynamic parallel constant-time algorithm that maintains \member{$D_k$}~on a "common CRCW PRAM" with $\mathcal{O}(W(n) \cdot \log n + (\log n)^3)$ processors, for each $k \geq 1$.
\end{lemma}

We use \Cref{lem:StringEqualityToDyck} to show that \member{$D_k$}~can be solved on $\mathcal{O}(n^{\epsilon})$ processors.
To this end, we need to show that we can maintain \stringEquality~efficiently.
Our solution for \stringEquality~uses the problem \rangeEval{$M,\circ$}, which we address first.

\subsection{Maintaining RangeEval}

The problem \rangeEval{$M,\circ$}~uses a ""monoid"" $(M,\circ)$, which consists of a set $M$, a binary associative operator $\circ: M \times M \to M$ and a neutral element $0 \in M$.
The problem \rangeEval{$M,\circ$}~maintains a sequence of elements $m_1 \ldots m_{n}$ and can determine, for any range $1 \leq l \leq r \leq n$, the result of $m_l \circ \ldots \circ m_r$.

\begin{itemize}
	\item \AP \emph{Problem:} \intro *\rangeEval{$M,\circ$}
	\item \emph{Initialization:} A sequence $m_1 \ldots m_n$ of "monoid" elements with $m_i = 0$ for each $i \in [1,n]$.
	\item \emph{Updates:}
	\begin{itemize}
		\item $\textsc{Set}(i,m)$: Replace $m_i$ with $m$.
	\end{itemize}
	\item \emph{Queries:}
	\begin{itemize}
		\item $\textsc{Range}(l,r)$: Return $m_l \circ \ldots \circ m_r$.
	\end{itemize}
\end{itemize}

Schmidt et al.\ show that \rangeEval{$M,\circ$}~can be maintained on $\mathcal{O}(n^{\epsilon})$ processors for finite "monoids" $(M,\circ)$ \cite[Proposition 5.2]{schmidtWorksensitiveDynamicComplexity2021}.
It is easy to see that their proof also works for other "monoids" $(M,\circ)$, so long as each $m_i \in M$ can be stored in a constant number of memory cells (i.e.\ $\mathcal{O}(\log n)$ bits) and the operator $\circ$ can be evaluated in constant time by a single processor.

\begin{lemma}\label{lem:RangeEval} \cite[Proposition 5.2]{schmidtWorksensitiveDynamicComplexity2021}
	Let $(M,\circ)$ be a "monoid" where each $m_i \in M$ can be stored in $\mathcal{O}(\log n)$ bits and the operator $\circ$ can be evaluated in constant time by a single processor.
	For any $\epsilon > 0$, there is a dynamic algorithm for \rangeEval{$M,\circ$}~on a "common CRCW PRAM", where
	\begin{itemize}
		\item the size of the data structure is in $\mathcal{O}(n)$,
		\item the initialization runs in parallel constant time on $\mathcal{O}(n)$ processors,
		\item updates run in parallel constant time on $\mathcal{O}(n^{\epsilon})$ processors, and
		\item queries run in parallel constant time on one processor.
	\end{itemize}
\end{lemma}

We do not repeat the proof here because it is identical to the original proof in \cite[Proposition 5.2]{schmidtWorksensitiveDynamicComplexity2021}.\footnote{
	We note that, if we take the algorithm from Schmidt et al.\ and replace the fixed tree with an adaptive tree like in \Cref{lem:dynString}, the modified algorithm even supports insertions and deletions.
	We do not discuss this modification in detail because it is unnecessary for our purposes.
}
Now that we have an algorithm for \rangeEval{$M,\circ$}, we can reduce \stringEquality~to \dynamicLCE.

\subsection{Maintaining String Equality}

To solve \stringEquality, we reduce \dynamicLCE~to it.
Recall that \dynamicLCE~is defined as follows.

\begin{itemize}
	\item \emph{Problem:} \dynamicLCE
	\item \emph{Initialization:} String $S=\epsilon$, alphabet $\Sigma = [1,n]$, maximum size $n$
	\item \emph{Updates:}
	\begin{itemize}
		\item $\textsc{Insert}(i,\sigma)$: Set $S$ to $S[1,i-1] \circ \sigma \circ S[i,|S|]$.
		\item $\textsc{Delete}(i)$: Set $S$ to $S[1,i-1] \circ S[i+1,|S|]$.
	\end{itemize}
	\item \emph{Queries:}
	\begin{itemize}
		\item $\textsc{Lcp}(i,j)$: Return $\max\{ m \mid S[i,i+m) = S[j,j+m) \}$.
		\item $\textsc{Lcs}(i,j)$: Return $\max\{ m \mid S(i-m,i] = S(j-m,j] \}$.
	\end{itemize}
\end{itemize}

According to \Cref{the:LCE}, we can solve \dynamicLCE~on $\mathcal{O}(n^{\epsilon})$ processors.
It remains to show that we can solve \stringEquality~on $\mathcal{O}(n^{\epsilon})$ processors as well.

\begin{lemma}\label{lem:stringEquality}
	For any $\epsilon > 0$, there is a dynamic algorithm for \stringEquality~on a "common CRCW PRAM", where
	\begin{itemize}
		\item the size of the data structure is in $\mathcal{O}(n \log n \log^* n)$,
		\item the initialization for $S = \bot^n$ runs in parallel constant time on $\mathcal{O}(n \log n \log^* n)$ processors,
		\item updates run in parallel constant time on $\mathcal{O}(n^{\epsilon})$ processors, and
		\item queries run in parallel constant time on $\mathcal{O}(n^{\epsilon})$ processors.
	\end{itemize}
\end{lemma}

We want to reduce \stringEquality~to \dynamicLCE.
It is not difficult to see that, if we let the string for \dynamicLCE~be the concatenation of the two strings from \stringEquality, answering queries becomes easy.
The challenge lies in the different updates.

\stringEquality~uses "void characters" $\bot$, whereas \dynamicLCE~has insertion and deletion.
While insertion and deletion is intuitively more powerful than "void characters", it is not obvious how to get from one to the other.
The difficulty lies in the fact that, for a string $S$ with "void characters", the position $S[i]$ in general does not correspond to $\erase{\bot}(S)[i]$.
Instead $S[i]$ corresponds to $\erase{\bot}(S)[j]$ where $j$ is the number of non-"void characters" in $S[1,i]$.

We use an instance of \rangeEval{$M,\circ$}~with an appropriate "monoid" $(M,\circ)$ to maintain which position in $S$ corresponds to which position in $\erase{\bot}(S)$.

\begin{proof}
	We show that we can use the algorithm for \dynamicLCE~from \Cref{the:LCE} to solve \stringEquality.
	
	Let $S_1,S_2$ be the strings from \stringEquality.
	The algorithm uses an instance of \dynamicLCE~with the concatenated string $\erase{\bot}(S_1 \circ S_2)$.
	It furthermore maintains the lengths $|\erase{\bot}(S_1)|, |\erase{\bot}(S_2)|$ of $\erase{\bot}(S_1)$ and $\erase{\bot}(S_2)$.
	
	The algorithm answers queries $\textsc{Equal}()$ by checking if $\textsc{Lcp}(1,|\erase{\bot}(S_1)|+1) = |\erase{\bot}(S_1)| = |\erase{\bot}(S_2)|$ holds.
	
	The algorithm uses instances of \rangeEval{$M,\circ$}~to maintain which position of $S_1$ corresponds to which position of $\erase{\bot}(S_1)$.
	We define a binary string $B_1$ that marks which positions in $S_1$ are not $\bot$.
	Let $B_1$ be a binary string with $|B_1| = |S_1| = n$ and $B_1[i] = 1$ if and only if $S_1[i] \neq \bot$.
	Then, for any index $i$ of $S_1$, it follows that $j = \sum_{h = 1}^{i}B[h]$ is the corresponding index in $\erase{\bot}(S_1)$.
	
	Let $M = [0,n]$ be the natural numbers from 0 to $n$ and let $\circ = +_{n+1}$ be addition modulo $n+1$.
	According to \Cref{lem:RangeEval}, \rangeEval{$[0,n],+_{n+1}$}~can be maintained on $\mathcal{O}(n^{\epsilon})$ processors.
	We use the binary string $B_1$ as the underlying sequence of "monoid"-elements for  \rangeEval{$[0,n],+_{n+1}$}.
	A query $\textsc{Range}(1,i) = \sum_{h = 1}^{i} B[i] = j$ then returns, for any index $i$ of $S$, the corresponding index $j$ in $\erase{\bot}(S)$.
	
	The algorithm processes an update $\textsc{Set}_1(i,\sigma)$ by first using the instance of\\ \rangeEval{$[0,n],+_{n+1}$}~to determine which index $j$ of $\erase{\bot}(S)$ corresponds to the first non-"void-character" before $S[i]$.
	It updates that instance \rangeEval{$[0,n],+_{n+1}$}~with $\textsc{Set}(i,1)$, increments $|\erase{\bot}(S_1)|$ by one, then updates \dynamicLCE~with $\textsc{Insert}(j+1,\sigma)$.
	
	Updates $\textsc{Reset}_1(i)$ are processed similarly.
	
	For updates to $S_2$ (updates $\textsc{Set}_2(i,\sigma)$ and $\textsc{Reset}_2(i)$), the algorithm uses a second instance of \rangeEval{$[0,n],+_{n+1}$}~to maintain which positions in $S_2$ correspond to which positions in $\erase{\bot}(S_2)$.
	Updates to $S_2$ are processed similarly to updates to $S_1$.
\end{proof}

We have shown that we can use our algorithm for \dynamicLCE~alongside an algorithm for \rangeEval{$M,\circ$}~to solve \stringEquality~efficiently.
\Cref{cor:dyck} follows from combining \Cref{lem:stringEquality} with \Cref{lem:StringEqualityToDyck}.

\section{Proof of \texorpdfstring{\Cref{cor:squares}}{Corollary 16}} \label{sec:proof:squares}

We show that a dynamic parallel constant-time algorithm can maintain \dynamicSquareFreeness~on a "common CRCW PRAM" with $\mathcal{O}(n^{\epsilon})$ processors.
This algorithm is derivative of an algorithm by Amir et al.\ \cite{amirRepetitionDetectionDynamic2019} modified to run in parallel constant time.
As with the original algorithm by Amir et al., this algorithm can be modified to report max "squares", enumerate all "squares" or enumerate all runs.
We only discuss the simplest incarnation for "square-freeness" here.

We begin by defining finite arithmetic progressions in \cref{subsec:sq:arithmetic}, which the algorithm uses to encode intermediate results efficiently.
Then we give an overview of the algorithm's structure in \cref{subsec:sq:structure}.
The algorithm runs over a series of stages, which we describe in \crefrange{subsec:sq:prefsuf}{subsec:sq:squareFree}.

\subsection{Finite Arithmetic Progressions} \label{subsec:sq:arithmetic}

The algorithm uses finite arithmetic progressions to encode intermediate results.

A finite arithmetic progression is a finite sequence of integers $i_1,i_2,\ldots,i_n$ where the difference between neighboring integers is constant, i.e.\ $i_j - {i_j+1} = i_{j+1} - i_{j+2}$ for all $j \in [1,n-2]$.
Finite arithmetic progressions can be written functions $P:[0,n-1] \to \mathbb{N}$ with $P(x) = a \cdot x + d$ with $i_1 = P(0), i_2 = P(1), \ldots, i_n = P(n-1)$.
This function is defined as a triple $(a,d,n)$ where $a$ is the first integer in the progression, $d$ is step between neighboring integers and $n$ is the length.
We encode finite arithmetic progressions as such triples $(a,d,n)$.

\subsection{Structure}\label{subsec:sq:structure}

The structure of this algorithm differs from the algorithms discussed so far.

The algorithm runs over a series of four stages.
Each stage consists of one or more instances of a dynamic algorithm.
Each instance uses one or more sub-instances of a dynamic algorithm from the stage below it.

Updates consist of a top-down step followed by a bottom-up step.
The top-down step passes the update through the stages to all instances that are affected, essentially computing which instances are affected and notifying them of the change.
The bottom-up step computes the effects from that change for each affected instance, starting at the lowest stage and traveling upwards.

The top-down step starts at the singular topmost instance.
Each instance that is updated passes this update along to all of its sub-instances that are affected until this chain reaches the lowermost instances.
The topmost instance passes the update to a logarithmic number of sub-instances.
Each following instance passes the update to a constant number of sub-instances.
The number of steps between the topmost and lowermost instances is constant.

The calculation during the bottom-up step for each instance is essentially static.
Each affected instance recalculates its result fully, then passes its result up to its parent-instance, so that it may recalculate its result.

The \emph{dynamic} aspect of this algorithm lies in the fact that each instance passes updates only to those sub-instances that are affected by it and sticks with the old results for unaffected sub-instances.

The topmost stage solves the problem \dynamicSquareFreeness.

\begin{itemize}
	\item \AP \emph{Problem:} \intro *\dynamicSquareFreeness
	\item \emph{Initialization:} String $S=\text{a}^n$, alphabet $\Sigma = \{a,\ldots\}$
	\item \emph{Updates:} $\textsc{Set}(x,\sigma)$: Set $S[x]$ to $\sigma \in \Sigma$.
	\item \emph{Query:} $\textsc{Query}()$: Does $S$ have a substring $uu, u \neq \epsilon$?
\end{itemize}

Note that, unlike the problems discussed before, this problem as well as the following problems do not support character insertion and do not support "void characters".

The algorithm solves this problem using $\mathcal{O}(n)$ instances of \dynamicRangeSquareFreeness, which restricts \dynamicSquareFreeness~to "squares" whose lengths and starting positions fall within certain ranges.

\begin{itemize}
	\item \AP \emph{Problem:} \intro *\dynamicRangeSquareFreeness
	\item \emph{Initialization:} String $S=\text{a}^n$, alphabet $\Sigma = \{a,\ldots\}$, integers $\ell, i$
	\item \emph{Updates:}  $\textsc{Set}(x,\sigma)$: Set $S[x]$ to $\sigma \in \Sigma$.
	\item \emph{Query:} $\textsc{Query}()$: Does $S$ have a substring $S[x,x+m] = uu$ with:
	\begin{itemize}
		\item  $x \in [\ell(i-2)+2,\ell(i-1)+2)$ and $m \in [4\ell, 6\ell)$.
	\end{itemize}
\end{itemize}

To tackle an instance of \dynamicRangeSquareFreeness, it uses one instance of \dynamicStringMatching.
The problem \dynamicStringMatching~uses a text $T$ and pattern $P$ where the text is exactly twice as long as the pattern.
It maintains every occurrence of the pattern in the text.
We will see that these occurrences can be encoded as one finite arithmetic progression.

\begin{itemize}
	\item \AP \emph{Problem:} \intro *\dynamicStringMatching
	\item \emph{Initialization:} Strings $T=\text{a}^{2n}$, $P=\text{a}^n$, alphabet $\Sigma = \{a,\ldots\}$
	\item \emph{Updates:}
	\begin{itemize}
		\item $\textsc{Set}_T(x,\sigma)$: Set $T[x]$ to $\sigma \in \Sigma$.
		\item $\textsc{Set}_P(x,\sigma)$: Set $P[x]$ to $\sigma \in \Sigma$.
	\end{itemize}
	\item \emph{Query:}
	\begin{itemize}
		\item $\textsc{Occ}()$: All occurrences of $P$ in $T$ encoded as a finite arithmetic progression.
	\end{itemize}
\end{itemize}

It resolves this problem with two instances of \dynamicPrefixSuffix.
The problem \dynamicPrefixSuffix~uses two strings $P$ and $S$ and maintains all prefixes of $P$ that are suffixes of $S$.
A query retrieves occurrences with lengths in $[2^{\ell}, 2^{\ell+1})$ because the occurrences in such a range can always be encoded as one finite arithmetic progression.

\begin{itemize}
	\item \AP \emph{Problem:} \intro *\dynamicPrefixSuffix
	\item \emph{Initialization:} Strings $P=\text{a}^{n}$, $S=\text{a}^n$, alphabet $\Sigma = \{a,\ldots\}$
	\item \emph{Updates:}
	\begin{itemize}
		\item $\textsc{Set}_P(x,\sigma)$: Set $T[x]$ to $\sigma \in \Sigma$.
		\item $\textsc{Set}_S(x,\sigma)$: Set $P[x]$ to $\sigma \in \Sigma$.
	\end{itemize}
	\item \emph{Query:}
	\begin{itemize}
		\item $\textsc{Match}(\ell)$: All $i \in [2^{\ell},2^{\ell+1})$ with $P[1,i] = S[|S|-i+1,|S|]$, encoded as a finite arithmetic progression
	\end{itemize}
\end{itemize}

Finally, the algorithm solves \dynamicPrefixSuffix~recursively over constant depth, where each instance uses $\mathcal{O}(n^{\epsilon})$ recursive sub-instances.

We describe these stages from the bottom up, starting with the recursive stage.

\begin{enumerate}
	\item Solve one instance of \dynamicPrefixSuffix~with $P,S$ with
	\begin{itemize}
		\item $\mathcal{O}(n^{\epsilon})$ instances of \dynamicPrefixSuffix~with $P',S'$, $n^{\epsilon}|P'S'| = |PS|$.
	\end{itemize}
	\item Solve one instance of \dynamicStringMatching~with $T,P$ with
	\begin{itemize}
		\item two instances of \dynamicPrefixSuffix~with $P',S'$, $|P'S'| = |T| = 2|S|$.
	\end{itemize}
	\item Solve one instance of \dynamicRangeSquareFreeness~with $S, \ell$ with
	\begin{itemize}
		\item One instance of \dynamicStringMatching~with $T',P'$, $|T'| = 2|P'| = 2\ell$.
	\end{itemize}
	\item Solve the topmost instance \dynamicSquareFreeness~with $S$ with
	\begin{itemize}
		\item $\mathcal{O}(n)$ instances of \dynamicRangeSquareFreeness~with $S' = S$.
	\end{itemize}
\end{enumerate}

We use the fact that we are only interested in "square-freeness" to simplify the stages.
If a stage has many candidates, it has a periodic substring which always contains a "square".
The algorithm only have to verify candidates if there a few of them, in which case it can afford to look at every candidate separately.

To this end, if any instance discovers a periodic substring, it may simply report that it found a "square", without actually solving its associated problem.

The algorithm also maintains an instance of \dynamicLCE~over the string $S$ to compare substrings efficiently.\footnote{
	We note that, alternatively, it is possible to adjust the algorithm used in \cite[Lemma 1]{amirRepetitionDetectionDynamic2019} to run in constant time on $\mathcal{O}(n^{\epsilon})$ processors.
	One can use \Cref{lem:RangeEval} to maintain Karp-Rabin fingerprints for each substring, where the elements in the "monoid" are tuples containing the fingerprint and length of a substring and the binary operator computes the fingerprint of the concatenation.
	
	We prefer to use \Cref{the:LCE} instead because Karp-Rabin fingerprints are randomized, may produce false positives, and maintaining them via \Cref{lem:RangeEval} allows only character substitutions.
}

\subsection{Recursive Dynamic Prefix-Suffix} \label{subsec:sq:prefsuf}

The first stage solves \dynamicPrefixSuffix~using $n^{\epsilon}$ recursive instances.
This algorithm is taken from Amir et al.\cite[Section 4.1]{amirRepetitionDetectionDynamic2019}, with some modifications to run in parallel constant time.

\begin{lemma}\label{the:sqfreeness:prefixsuffix}
	For any $\epsilon > 0$, there is a dynamic constant-time algorithm for \dynamicPrefixSuffix~on a "common CRCW PRAM" with $\mathcal{O}(n^{\epsilon})$ work.
	If $P$ or $S$ contains a periodic substring, a query may return this information instead of the regular result.
\end{lemma}

We use $n$ when referring to the length of $P$ and $S$ in the original instance at the highest level of the recursion.
We use $m$ for the length of $P$ and $S$ in recursive instances.

\subparagraph{Result Representation}
We first describe how to represent the result efficiently.
While there may be up to $O(m)$ matches (as can be seen with $P = S = a^m$), we can always store these matches as just $O(\log m)$ finite arithmetic progressions.
This relies on the following observation.
\begin{observation}\label{lem:prefix-suffix-result-periodicity}
	Let $P,S \in \Sigma^m$ be some input for \dynamicPrefixSuffix~and let $i \lneq j$ be two matches such that $P[1,i] = S[m-i+1,m]$ and $P[1,j] = S[m-j+1,m]$.
	
	Then $P[1,j] = S[m-j+1,m]$ has a period of length $j-i$.
	
	If there are no matches between $i$ and $j$, then $j-i$ is the length of the shortest period and the matches between $j$ and $(j \text{ mod } (j-i))$ are precisely all elements of the arithmetic progression $P(x) = j - (j-i)x$ for $0 \leq x < \lfloor \frac{j}{j-i} \rfloor+1$.
\end{observation}
A match is an integer from $[1,m]$.
We divide this interval into segments of doubling length $I_1 = [1,1], I_2 = [2,3], I_3 = [4,7], \ldots$, or in general $I_i = [2^{i-1}, 2^i-1]$.
\begin{proposition}
	All matches within the same segment $I_i$ follow a single finite arithmetic progression.
\end{proposition}
\begin{proof}
	Let $I_i$ be an arbitrary segment.
	If there are two or fewer matches, it is trivial to express them as a finite arithmetic progression.
	We therefore assume that there are at least three matches in $I_i$.
	Let $a \gneq b$ be the last and second-to-last match in $I_i$.
	From $a, b \in I_i = [2^{i-1}, 2^i-1]$, it follows that $a-b < 2^{i-1}$.
	According to \Cref{lem:prefix-suffix-result-periodicity}, all matches in $[a, a \text{ mod } a-b]$ follow a finite arithmetic progression.
	This covers all matches in $I_i$, since $a-b < 2^{i-1}$ implies $a \text{ mod } a-b < 2^{i-1}$ and because there are by definition no matches in $I_i$ after $a$.
\end{proof}
It suffices to store just one finite arithmetic progression for each of the $\mathcal{O}(\log m)$ segments.
We call these segments the result segments.

\subparagraph{Algorithm}
The algorithm for \dynamicPrefixSuffix~divides its input strings in $n^{\epsilon}$ blocks $P = P_1P_2\ldots P_{n^{\epsilon}}$ and $S = S_{n^{\epsilon}}\ldots S_2S_1$ of equal lengths.
For every $i \in \{ 1, \ldots, \frac{n^{\epsilon}}{2} \}$, it uses two recursive instances of \dynamicPrefixSuffix~- one with $(P_i, S_i)$ and another with $(P_i,S_{i+1})$.

The algorithm is based on the observation that a match in the original instance implies a match in one of its recursive instances.
\begin{observation}
	Let $P,S$ be strings with $|P| = |S| = m$ and a match $a$ such that $P[1,a] = S[m-a+1,m]$.
	
	Then, for all substrings $P' = P[l_p,r_p]$ and $S' = S[l_s,r_s]$ with $|P'| = |S'| = m'$, $P',S'$ have match $a' = a-\ell_p-m+r_s+1$ such that $P'[1,a'] = S[m'-a'+1,m']$, if $a' \in [1,m']$.
\end{observation}
An update that affects some block $P_i$ or $S_i$ is passed to the recursive instances that use this block.
There are up to two instances that use any given block.

The algorithm then recalculates all its matches by looking at the finite arithmetic progression $(a',d',n')$ in each result segment of a recursive instance \dynamicPrefixSuffix~with $P' = P[l_p,r_p]$ and $S' = S[l_s,r_s]$.
There are three cases, based on the number of matches.

\begin{itemize}
	\item $n' = 0$: There are no candidates and no matches.
	\item $n' = 1$: Verify if the sole candidate $a$ is a match.
	\begin{itemize}
		\item Verify with one query to \dynamicLCE.
	\end{itemize}
	\item $n > 1$: There is a periodic substring and therefore a "square".
\end{itemize}

If the algorithm encountered a finite arithmetic progression with $n'>1$, it does not have to calculate a result (apart from reporting that it found a "square") and terminates.
Otherwise, it calculates a finite arithmetic progression for every result segment by finding the first, second and last match in this segment.

One instance \dynamicPrefixSuffix, not counting recursive calls, runs in constant time with $\mathcal{O}(m^{\epsilon})$ work.
An update affects up to two recursive instances, while queries simply return the finite arithmetic progression which was calculated and stored during the update.

With each level of the recursion, the size of the inputs is reduced by a factor of $n^{\epsilon}$.
The recursion reaches inputs of constant size after $\frac{1}{\epsilon}$ levels, which is constant.

It follows that the entire algorithm runs in constant time with $\mathcal{O}(n^{\epsilon})$ work.

\subsection{Dynamic Prefix-Suffix to Dynamic String-Matching} \label{subsec:sq:stringMatch}

The second stage solves \dynamicStringMatching~using two instances of \dynamicPrefixSuffix.
The idea behind this algorithm appears in Amir et al.\cite[Section 4]{amirRepetitionDetectionDynamic2019}.

\begin{lemma}\label{the:sqfreeness:stringmatching}
	For any $\epsilon > 0$, there is a dynamic constant-time algorithm for \dynamicStringMatching~on a "common CRCW PRAM" with $\mathcal{O}(n^{\epsilon})$ work.
	If $P$ or $S$ contains a periodic substring, a query may return this information instead of the regular result.
\end{lemma}

\subparagraph{Result Representation}
This algorithm will only ever report a single match.
If there are two or more matches, they have to touch or overlap, creating a periodic substring.
A periodic substring always contains a "square", meaning that the algorithm does not have to calculate a result.

\subparagraph{Algorithm}
The algorithm for \dynamicStringMatching~with $(T,P)$ divides $T$ in two halves $T_1T_2 = T$.
It uses two instances \dynamicPrefixSuffix, one with $(P,T_1)$ and another with $(T_2,P)$.
If any instance of \dynamicPrefixSuffix~reports a "square", \dynamicStringMatching~simply passes this information along.
In the following, we assume that neither instance found a "square" and instead returned a finite arithmetic progression for each of its result segments.

The algorithm is based on the observation that an occurrence of $P$ in $T$ corresponds to one prefix-suffix-match between $P$ and $T_1$ and an opposite match in between $T_2$ and $P$.
\begin{observation}
	There is an $i$ with $T[i,i+|P|) = P$ if and only if $P[1,|P|-i+1] = T_1[i,|P|]$ and $T_2[1,i-1] = P[|P|-i+2,|P|]$.
\end{observation}
The algorithm uses one instance of \dynamicPrefixSuffix~to identify candidates and the other to verify them.

For every result segment of the first instance, it looks at its finite arithmetic progression $(a',d',n')$.
There are two cases.
\begin{itemize}
	\item $n' = 0$: There are no candidates and no matches.
	\item $n' = 1$: Verify if the sole candidate $a'$ is a match by testing if there is an opposite match in the second instance.
	\item $n' > 1$: There is a periodic substring and therefore a "square".
\end{itemize}
If the algorithm encountered a finite arithmetic with $n'>1$, it can simply report that it found a "square" and terminate.
If it found a match stemming from two or more result segments, it can also report that it found a "square" and terminate.
Otherwise, it reports the only match it found, if it exists.

The algorithm uses two instances of \dynamicPrefixSuffix, which run in constant time with $\mathcal{O}(n^{\epsilon})$ work according to \Cref{the:sqfreeness:prefixsuffix}.
The rest of the algorithm also runs in constant time with $\mathcal{O}(n^{\epsilon})$ work.

\subsection{Dynamic String-Matching to Dynamic Range-Square-Freeness} \label{subsec:sq:rangeSquare}

The third stage solves \dynamicRangeSquareFreeness~using one instance of \dynamicStringMatching.
For the two topmost stages, we switch to the static "squares" algorithm by Apostolico and Breslauer \cite{apostolicoOptimalOlogLog1996}, because this algorithm already uses a parallel machine.
One instance of \dynamicRangeSquareFreeness~stage corresponds to what Apostolico and Breslauer call a \textit{substage} \cite[Section 3.1]{apostolicoOptimalOlogLog1996}.

\begin{lemma}\label{the:sqfreeness:rangesquare}
	For any $\epsilon > 0$, there is a dynamic constant-time algorithm for \dynamicRangeSquareFreeness~on a "common CRCW PRAM" with $\mathcal{O}(n^{\epsilon})$ work.
\end{lemma}

\subparagraph{Result Representation}
An instance stores a single bit, set to 1 if and only if it finds a "square".

\subparagraph{Algorithm}
Recall that \dynamicRangeSquareFreeness~has two parameters $\ell,i$, set during the initialization, and searches for "squares" of length in $[4\ell,6l]$ and with starting index in $[l(i-2)+2,l(i-1)+1]$.

The algorithm for \dynamicRangeSquareFreeness~divides $S$ in blocks $\mathcal{B}_{\ell,i}$ of length $\ell$, where $\mathcal{B}_{\ell,i}$ starts at index $\ell(i-1)+1$.
This instance is only concerned with the series of blocks $\mathcal{B}_{\ell,i-1},\mathcal{B}_{\ell,i},\ldots,\mathcal{B}_{\ell,i+5}$.

It uses one instance of \dynamicStringMatching~with $T = \mathcal{B}_{\ell,i+2}\mathcal{B}_{\ell,i+3}$ and $P = \mathcal{B}_{\ell,i}$.
The algorithm passes changes to $\mathcal{B}_{\ell,i}$, $\mathcal{B}_{\ell,i+2}$ or $\mathcal{B}_{\ell,i+3}$ along to \dynamicStringMatching.
If this instance reports a "square", \dynamicRangeSquareFreeness~simply passes this along with no further calculations.
In the following, we assume that \dynamicStringMatching~did not report a "square" and instead returns a single match $a'$.

The algorithm is based on the observation that all "squares" this instance is interested in fully contain $\mathcal{B}_{\ell,i}$ in their left side.
\begin{observation}
	If $S[x,x+m)$ is a "square" with $m \in [4\ell,6l]$ and $x \in [l(i-2)+2,l(i-1)+1]$, then $S[x, x+\frac{m}{2})$ fully contains $\mathcal{B}_{\ell,i}$ and $S[x+\frac{m}{2},x+m)$ contains a copy of $\mathcal{B}_{\ell,i}$ with offset $\frac{m}{2} \in [2\ell,3l]$.
\end{observation}
The algorithm uses \dynamicStringMatching~to find the occurrence of $\mathcal{B}_{\ell,i}$ in $S[x+\frac{m}{2},x+m)$.
The copy lies in $\mathcal{B}_{\ell,i+2}\mathcal{B}_{\ell,i+3}$, which is precisely the text for the instance of \dynamicStringMatching.

The algorithm verifies if the single match reported by \dynamicStringMatching~can be extended to become a "square" using \dynamicLCE.

\subsection{Dynamic Range-Square-Freeness to Dynamic Square-Freeness} \label{subsec:sq:squareFree}

The final stage solves \dynamicSquareFreeness~using $\mathcal{O}(n)$ instances of \dynamicRangeSquareFreeness.
The idea of decomposing one instance \dynamicSquareFreeness~into many instances of \dynamicRangeSquareFreeness~is taken from Apostolico and Breslauer \cite[Section 3]{apostolicoOptimalOlogLog1996}.

\subparagraph{Result Representation}

The algorithm stores a single bit for each instance of \dynamicRangeSquareFreeness~where the bit is $0$ if and only if that instance does not report a "square".
The algorithm uses an instance of \dynamicNeighbor~to store and query these bits efficiently.

\subparagraph{Algorithm}
The algorithm uses $\log_{\frac{3}{2}} n$ parallel levels.
Level $k$ is comprised of instances of \dynamicRangeSquareFreeness~with $\ell = (\frac{3}{2})^k$ and $i \in [1, n(\frac{2}{3})^k]$.
Updates are passed along to every instance of \dynamicRangeSquareFreeness~which are affected.
Recall that \dynamicRangeSquareFreeness~divides $S$ in blocks $\mathcal{B}_{\ell,i}$ of length $\ell$ and is only interested in the seven blocks $\mathcal{B}_{\ell,i-1},\mathcal{B}_{\ell,i},\ldots,\mathcal{B}_{\ell,i+5}$.
It follows that any update to a single character affects up to seven instances per level, or $\mathcal{O}(\log n)$ instances in total.

After all affected instances of \dynamicRangeSquareFreeness~have recalculated their results, the algorithm updates the bits for \dynamicNeighbor~accordingly.

A query uses \dynamicNeighbor~to determine if all bits are $0$ and reports that $S$ is "square-free" if that is the case.

\subsection{Summary}

We have shown that we can modify the algorithm from Amir et al.\ \cite{amirRepetitionDetectionDynamic2019} to run in parallel constant time, using our algorithm for \dynamicLCE.
We have only shown the modifications for \dynamicSquareFreeness.
It is also possible to maintain the maximum "square", all "squares", and all runs with extensions of this algorithm, as shown by Amir et al.\ for their sequential algorithm.

\section{Overcoming the Size Limitation}\label{sec:SizeLimit}

Our algorithm has the limitation that a maximum size for the string is fixed during the initialization.
We give a rough description of a modified algorithm without this limitation using ideas from the Muddling Lemma \cite{dattaStrategyDynamicPrograms2019}.

The modified algorithm uses two instances of the original algorithm.
It uses the ""current instance"" $\mathcal{I}$ with some maximum size $m$ that is larger then the current string.
The ""future instance"" $\mathcal{I}'$ has maximum size $2m$ and is constructed as the string grows larger.

Simply speaking, the algorithm uses the "current instance" $\mathcal{I}$ to answer queries, while in the background building the "future instance" $\mathcal{I}'$.
Both instances represent the same string, but the "future instance" is (in general) not yet fully built.
It is fully built by the time the current string grows to large for the "current instance" $\mathcal{I}$, at which point the algorithm exchanges $\mathcal{I}$ for $\mathcal{I}'$, then starts building a new "future instance" $\mathcal{I}'$ with twice the maximum size.

Building the "future instance" $\mathcal{I}'$ is split across the updates.
The algorithm starts building the "future instance" $\mathcal{I}'$ when the current string has length $\frac{m}{2}$, then swaps over when it reaches length $m$.
It therefore builds $\mathcal{I}'$ over the course of at least $\frac{m}{2}$ updates, with $\mathcal{O}(m^{\epsilon})$ processors per update.
It is possible to build $\mathcal{I}'$ with maximum size $m$ for any underlying string in $\mathcal{O}(m)$ time on $\mathcal{O}(m^{\epsilon})$ processors.

One complication with this approach is that the underlying string for the "future instance" $\mathcal{I}'$ changes as it is built.
To overcome this, each update is also applied to the "future instance" $\mathcal{I}'$, even though it may not be fully built yet.
Effects of the update to parts of $\mathcal{I}'$ that have not been built yet are simply ignored, while other effects are applied as usual.

\end{document}